\documentclass[journal]{IEEEtran}
\usepackage[T1]{fontenc}
\usepackage{times}
\usepackage{amsmath}
\usepackage{amsthm}
\usepackage{color}
\usepackage{float}
\usepackage{array}
\usepackage{amssymb}
\usepackage{amsfonts}
\usepackage{graphicx}
\usepackage[sans]{dsfont}
\usepackage{verbatim} 
\usepackage{multirow} 
\usepackage{slashed,graphicx}
\usepackage{url}
\usepackage{subcaption}
\usepackage{tablefootnote}
\usepackage{color}
\usepackage{xcolor}
\usepackage{cite}
\usepackage{tabulary}
\usepackage{booktabs}
\usepackage{mathtools}
\usepackage{textcomp}
\usepackage[boxed, commentsnumbered, ruled,vlined]{algorithm2e}

\usepackage[hang,flushmargin]{footmisc}
\usepackage{lipsum}
\makeatletter
\newcommand{\algorithmfootnote}[2][\footnotesize]{%
  \let\old@algocf@finish\@algocf@finish
  \def\@algocf@finish{\old@algocf@finish
    \leavevmode\rlap{\begin{minipage}{\linewidth}
    #1#2
    \end{minipage}}%
  }%
}
\makeatother

\usepackage{environ}
\NewEnviron{scaled-equation}{%
\begin{equation}
\scalebox{0.95}{$\BODY$}
\end{equation}
}

\theoremstyle{plain}
\newtheorem{mydef}{Definition}
\newtheorem{mylemma}{Lemma}
\newtheorem{mytheorem}{Theorem}

\makeatletter
    \def\tagform@#1{\maketag@@@{\normalsize(#1)\@@italiccorr}}
\makeatother

\ifodd 1

\newcommand{\rev}[1]{{{\color{black} #1}}}
\newcommand{\revv}[1]{{{\color{black} #1}}}
\else
\newcommand{\rev}[1]{#1}
\newcommand{\revv}[1]{#1}
\newcommand{\com}[1]{}
\fi

\begin{document}
\title{Pricing Mechanisms for Crowd-Sensed Spatial-Statistics-Based Radio Mapping\\
}
\author{Xuhang~Ying,~\IEEEmembership{Student Member,~IEEE,}
        Sumit~Roy,~\IEEEmembership{Fellow,~IEEE,}
        and~Radha~Poovendran,~\IEEEmembership{Fellow,~IEEE}

\thanks{
Manuscript received October 10, 2016;
revised March 13, 2017; accepted April 28, 2017.
Date of current version April 28, 2017. 
This work was supported in part by NSF EARS Award 1443923 and NSF CPS award CNS-1446866.
}
\thanks{The authors are with the Department of Electrical Engineering, University of Washington, Seattle WA, 98195. Email: \{xhying, sroy, rp3\}@uw.edu.
This work was present in part at IEEE GLOBECOM 2016 \cite{ying2016pricing}.
}
}

\maketitle

\begin{abstract}
Networking on white spaces (i.e., locally unused spectrum) relies on active monitoring of spectrum usage. 
Spectrum databases based on empirical radio propagation models are widely adopted but shown to be error-prone, since they do not account for built environments (e.g., trees and man-made buildings). 
As an alternative, crowd-sensed radio mapping by mobile clients who acquire local spectrum data and transmit it to a central aggregator (platform) for processing, results in more accurate radio maps. 
Success of such crowd-sensing systems presumes some incentive mechanisms to attract user participation. 
In this work, \rev{we assume that the platform who constructs radio environment maps makes one-time offers (the incentive for participation) to users, and collects data from those who accept the offers.
We design pricing mechanisms based on \textit{expected utility (EU) maximization}, where EU captures the tradeoff between radio mapping performance (location and data quality), crowd-sensing cost and uncertainty in offer outcomes (i.e., possible expiration and rejection).
Specifically, we consider both \textit{sequential offering}, where one best price offer is sent to the best user in each round, and \textit{batched offering}, where a batch of offers is made in each round.
For the later, we show that EU is \textit{submodular} in the discrete domain, and propose a mechanism that first fixes the pricing rule, and selects users based on \textit{Unconstrained Submodular Maximization} (USM); it then compares different pricing rules to find the best batch of offers in each round. We show that USM-based user selection has provable performance guarantee. 
Proposed mechanisms are evaluated and compared against utility-maximization-based baseline mechanisms.
}

\end{abstract}

\begin{IEEEkeywords}
Radio Environment Mapping, Spatial Statistics, Crowd-Sensing, Pricing Mechanism, Expected Utility Maximization, Unconstrained Submodular Maximization.
\end{IEEEkeywords}




\section{Introduction}
The exponential increase in mobile data traffic has naturally translated into similar demands for wireless network capacity for broadband access. 
A significant portion of spectrum is allocated via licensing for various services (e.g., TV broadcasting, radar and satellite services), but is often under-utilized in practice.
To improve spectrum utilization, the FCC allows unlicensed users to opportunistically access locally idle licensed spectrum, aka \textit{White Spaces} (WS), subject to the no-harmful-interference constraint \cite{fcc2008second}.

A key step of WS networking is to actively monitor spectrum usage locally and identify WS opportunities.
Currently, empirical radio propagation models are widely used and implmented in spectrum databases \cite{harrison2010much,makris2012quantifying,fcc2010second,hessar2015capacity}  for this purpose, but recent studies \cite{phillips2011bounding, achtzehn2014improving, ying2015revisiting} have shown that they are often locally inaccurate since they do not account for built environments like trees and man-made buildings.
\rev{To augment spectrum databases}, radio mapping via \textit{spatial statistics} (e.g., Kriging \cite{cressie2015statistics, ying2015revisiting, achtzehn2012improving} and Gaussian Process (GP) \cite{deshpande2004model,krause2008near}) has been proposed, which leverages local measured RSSI\footnote{Received Signal Strength Indicator} data and provides accurate RSSI estimation at unmeasured locations.

RSSI data collection requires campaigns like drive tests, which are time- and labor-intensive for monitoring a wide area.
Deploying specialized spectrum sensors is a second option, but large-scale deployment  is typically costly in proportion to the hardware and deployment expenses. 
An economically viable alternative is \textit{crowd-sensing} with commodity mobile devices that are enabled to do such spectrum sensing, i.e., outsourcing sensing tasks opportunistically to mobile users with clients embedded with such sensing devices. 
However, since users consume resources (e.g., battery, CPU and memory) for sensing, they need to be properly compensated or \textit{incentivized}. 

\rev{Crowd-sensed radio mapping is different from other crowd-sensing applications \cite{yang2012crowdsourcing, koutsopoulos2013optimal, peng2015pay} in several ways.
First, it samples the (unknown) RSSI field (opportunistically) at reported user locations at a time instant (or within a short duration), and applies spatial statistics to estimate RSSI values at unmeasured locations. 
Therefore, to obtain a RSSI estimate at a {\em target} location that is not sampled, only the {\em relative positioning} of selected users (and the local spatial statistical model used for interpolation) matters.
Second, user devices are heterogeneous  (e.g. various smartphones and tablets from different manufacturers and even different model families from the same manufacturer) and naturally provide data of different quality due to manufacturing variances and specifically, noise figure of the underlying circuit design, which needs to be accounted in our spatial interpolation. 
Third, user devices are not dedicated to sensing but share resources (CPU cycles and battery notably) with  many other tasks.
Hence, in addition to energy or battery costs, users will incur \textit{opportunity costs} depending on current device statuses, which is the loss of utility if they decide to spend resources on sensing instead of other tasks. } 

In this work, we consider \textit{pricing} \cite{singla2013truthful,he2014toward,han2016posted} for crowd-sensed radio mapping. 
Given a spatial statistical model, the platform determines the \textit{value} of a set of users based on location, data quality and its own preferences, which is the amount (of money) it is willing to pay. 
Each user has a private \textit{sensing cost} (i.e.,  sum of energy and opportunity costs).
Selected users receive one-time price offers and have only one chance to make a decision (either accept or reject) by a given deadline; otherwise, the offer will be expired. 
Here, we consider \textit{rational} users who accept offers if its sensing cost is no greater than the offered price.
Those who accept offers perform sensing at reported locations, upload data and receive  \textit{payments}.
Therefore, from the platform's perspective, a set of users is associated with a \textit{utility} or gain, which is the difference between value and total payment.
Due to possible offer expiration (due to network congestion etc.) and rejection, the platform aims to maximize the \textit{expected utility} (EU) by determining who to send offers to (i.e., \textit{user selection}) and how much to offer (i.e., \textit{price  determination}). 


Our primary contributions are as follows:
\begin{itemize}
\item We design a crowd-sensing system that periodically acquires spectrum data from users for radio mapping.

\item We introduce EU and formulate pricing mechanism design as \textit{EU maximization}. 
We first propose \textit{sequential offering}, where the platform sends out the best offer to the best user in each round, and keeps offering until the next one is no longer profitable.
Then we generalize it to \textit{batched (i.e., single-batch and multi-batch) offering}, where a batch of multiple offers are made in each round. 

\item For batched offering, \rev{we show that EU is \textit{submodular} in the discrete domain.
We propose a pricing mechanism that first fixes the pricing rule, and selects users based on \textit{Unconstrained Submodular Maximization} (USM); it compares different pricing rules to find the best batch of offers that maximizes EU (instead of the best-case utility) in each round. 
\revv{We adopt the linear-time deterministic USM algorithm that provides a $1/3$-approximation guarantee \cite{buchbinder2015tight} for user selection.
In practice, however, EU is difficult to analytically evaluate and Monte-Carlo estimated EU is fed to the algorithm. 
We show that its worst-case performance is degraded by estimation errors, and the reduced amount grows linearly in the number of users given the maximum estimation error (Theorem~\ref{theorem:USM_with_approx_EU_bound}).}}

\item \rev{We conduct simulations to evaluate the proposed EU-maximization-based mechanisms, and further compare them against baseline mechanisms \revv{that aim to find the best batch of offers that maximize the best-case utility  in each round}.
Results show that our single-batch mechanism is better than the single-batch baseline mechanism with an improvement ranging from $8.5\%$ to $40.5\%$.
If more batches are allowed, our multi-batch mechanism achieves close performance with the multi-batch baseline mechanism, but requires much fewer batches ($2.5$ versus $7.7$ batches on average) and thus a much smaller delay.
Sequential offering works better than the single-batch baseline mechanism, but has a very large cumulative delay. 
Offer expiration adversely affected all mechanisms, but sequential and multi-batch offering are more robust.}
\end{itemize}

The rest of this paper is organized as follows.
We review related work in Section~\ref{sec:related_works} and provide a two-user tutorial example in Section~\ref{sec:example}.
In Section~\ref{sec:model}, we provide background on submodularity and present our models.
Our pricing mechanisms are presented in Section~\ref{sec:pricing_mechanism} and evaluated in Section~\ref{sec:evaluation}.
We conclude this study in Section~\ref{sec:conclusion}.

\section{Related Works} \label{sec:related_works}
In recent years, spatial-statistics-based radio mapping has been proposed to better capture local radio environments to augment spectrum databases. 
In \cite{phillips2012practical}, Phillips \textit{et al.} applied a statistical interpolation technique called Ordinary Kriging to map the coverage of WiMax networks. 
Similar techniques have been applied to estimate the coverage area of single-transmitter \cite{ying2015revisiting} and multi-transmitter networks\cite{achtzehn2012improving} in TV bands.
A more detailed discussion is available in \cite{yilmaz2013radio}. 

Radio mapping requires a large amount of sensing data, and incentivized crowd-sensing is considered as an economically viable option.
A number of various incentive mechanisms have been proposed.
In \cite{yang2012crowdsourcing}, Yang \textit{et al.} studied a platform-centric incentive model, where users share the reward proportionally in a Stackelberg game.
In mechanisms based on reserve auction \cite{yang2012crowdsourcing, feng2014trac, ying2015incentivizing}, users bid for tasks and receive payments no less than bids when selected. 
One main goal for the platform is to design a truthful mechanism that motivates users to bid at their true private costs.
In \cite{koutsopoulos2013optimal}, Koutsopoulos designed an incentive mechanism to determine participation level and payment allocation to minimize platform's compensation cost with guaranteed service quality.
Other models include all-pay auction \cite{luo2016incentive}, Bayesian models \cite{koutsopoulos2013optimal}, Tullock contests \cite{luo2015crowdsourcing} and posted pricing \cite{singer2013pricing, singla2013truthful, han2016posted}. 
Some are proposed in an \revv{\textit{online}} setting with constraints like budget limits \cite{singer2013pricing, zhao2014crowdsource, han2016taming}, where users arrive in a random order and a typical goal is to maximize a certain objective (e.g., revenue).

\rev{Incentive mechanisms are typically tailored to the crowd-sensing application being considered by incorporating factors like user location, data quality and user availability etc.
As an example, in \cite{feng2014trac}, each task has a specific location tag and each user can only compete for tasks within its service region.
In \cite{peng2015pay}, Peng \textit{et al.} extended the well-known Expectation Maximization algorithm to estimate the quality of sensing data and incorporated it in determining rewards.
In \cite{han2016posted}, Han \textit{et al.} studied a quality-aware Bayesian pricing problem where both users' sensing costs and qualities are random variables, drawn from known distributions.
The goal is to choose an appropriate posted price to recruit a group of users with reasonable sensing quality, and minimize the total expected payment.
If users need to move to designated sensing locations or are available at different time periods, then incentive mechanism design is closely coupled with task allocation \cite{he2014toward} or scheduling \cite{han2016truthful}.}

In this study, we consider incentive mechanism design in an \textit{offline} setting for crowd-sensed radio mapping, \revv{in which the platform acquires data from a pool of users who are interested and available for sensing in each period, and measurements are taken at their current locations}.
We consider data quality in terms of hardware quality, and incorporate it into the spatial statistical model (i.e., GP). 
Distinct from the auction-based incentive mechanism for crowd-sensed radio mapping \cite{ying2015incentivizing}, we are interested in {\em pricing mechanisms where the platform makes one-time price offers to a set of selected users}, and collects data from those who accept offers.
To select users and determine corresponding price offers, we define utility for the platform to trade the value it obtains from the resulting radio map generated based on the offered data, against the total price (crowd-sensing cost), and use the notion of  EU to account for possible offer expiration and rejection. 
We formulate the pricing mechanism design as EU maximization, and propose mechanisms based on USM. 
\section{A Two-User Tutorial Example}\label{sec:example}
\revv{In this section}, we \rev{first present our system architecture and} provide a two-user tutorial example to illustrate the basic idea of pricing for crowd-sensed radio mapping.

\subsection{\rev{System Architecture}}\label{sec:sys_arch}
As shown in Fig.~\ref{fig:arch_pricing}, the platform acquires data periodically from users.
At the beginning of each period, the platform broadcasts a sensing task to all users in the area of interest (AoI) with specific sensing parameters (e.g., center frequency, sampling rate and FFT bin size) to ensure a consistent sensing procedure across different hardware. 
\rev{Note that the task does not specify sensing locations for two reasons.
\revv{First, there is no need since the platform will take a sampling approach, that is, selecting a subset of users {\em after they provide their locations}. 
This is consistent with the underlying spatial statistical model, where only the relative positioning (instead of absolute locations) that matters for the resulting radio map quality.}
Second, it requires extra time and costs for users to move to target sensing locations, which means extra incentivization costs for the platform and added complexity for mechanism design\footnote{See \cite{he2014toward} for more discussions on allocation of tasks with specific locations.}}. \rev{To avoid excessive delay due to communication delay or failure, network congestion 
etc., each offer has a deadline, by which a decision has to be received by the platform (along with the data if accepted); otherwise, the offer will be expired.}

\begin{figure}[t]
    \centering
    \includegraphics[width = 1\columnwidth]{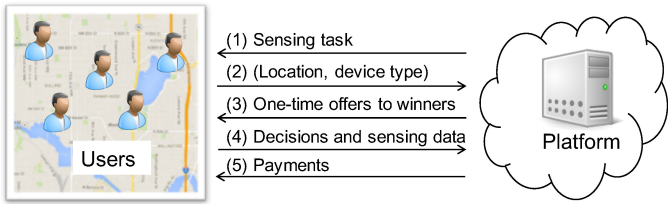}
    \caption{\rev{Pricing-based crowd-sensing system.
    (1) In each period, the platform first broadcasts a sensing task to users in the area of interest.
    (2) Users who are interested and available for sensing in the current period report location and device type.
    (3) The platform determines and sends out one-time offers to selected users.
    (4) If a selected user decides to accept the offer, it performs the required sensing task. 
    All selected users will inform the platform of their decisions (and upload the data) before the deadline. 
    (5) The platform pays users who contribute data.
    }
    }
    \label{fig:arch_pricing}
\end{figure}

In this study, we assume no entry or other overhead costs, that is, a user does not incur a fee to communicate with the platform.
We consider users of low mobility (e.g., pedestrians), who are honest in providing their information and following the protocol. 
We assume small displacements between reported locations and eventual sensing locations. 
We will leave the high-mobility case and security considerations as future work.

\subsection{\rev{A Two-User Scenario}}
Fig.~\ref{fig:two_user_example_topology} illustrates the topology of the two-user example.
The goal of the platform is to estimate the RSSI $Z(x)$ (in dBm) at each location $x \in U$, where $U$ represents the discretized AoI.
There are two users $S=\{1,2\}$ at $x_1$ and $x_2$ in the AoI. 
In each period, each user will incur a \textit{sensing cost} $c_i > 0$ and receive an \textit{offer} $p_i > 0$, when selected by the platform.
We assume \textit{rational} users, who accept the offer if $c_i\leq p_i$, and reject it otherwise.
\rev{In the following discussion, we assume no expired offers and will consider them later in Section~\ref{sec:pricing_mechanism}.}

\begin{figure}[ht!]
    \centering
    \includegraphics[width = 0.5\columnwidth]{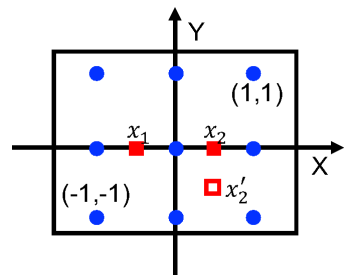}
    \caption{Topology of the two-user example. The AoI is discretized into a mesh grid of $9$ points (blue dots). User $1$ is at $(-0.5, 0)$, and user $2$ is at $(0.5,0)$ or $(0.5, -0.5)$.}
    \label{fig:two_user_example_topology}
\end{figure}

The platform has a \textit{valuation} function $v:2^S\mapsto \mathbb{R}_+$ and a \textit{pricing} function $p:2^S\mapsto \mathbb{R}_+$.
For each set of users $A$, there is an associated \textit{value} $v(A)$ and a total price of all offers $p(A)=\sum_{i\in A}p_i$ (i.e., crowd-sensing cost for the platform), assuming that each offer is unexpired and accepted.
It makes sense for the platform, \rev{as a rational decision maker}, to maximize its \textit{utility} (or profit), i.e.,
\begin{equation}
\max_{A\subseteq S} u(A) = \max_{A\subseteq S} \left( v(A) - p(A) \right),
\label{eq:goal_max_utility}
\end{equation}
\rev{which is an important and widely used concept in economics (e.g., rational choice theory \cite{scott2000rational}) and measures the platform's preference over the set of users $A$.}
For convenience, we also write $v(A)$, $p(A)$ and $u(A)$ as $v_A$, $p_A$ and $u_A$, respectively.

\rev{In practice, however, the probability that an offer is accepted (and data is uploaded to the platform) by the deadline is less than 1, due to possible expiration and rejection.
Hence, only a subset of users in $A$ accept offers, and $u(A)$ is essentially the \textit{best-case} utility.
In this case, it makes more senses for the platform to maximize the average-case or \textit{expected utility} (EU).
More will be discussed later in this example.}


\subsection{How to Valuate \rev{Users}}
In this example, two RSSI data models are considered: \rev{one ignoring shadowing and noise, and the other one \revv{is our GP-based model that accounts for both}}. 
\rev{Both models assume that any small-scale fading over small distances and time has been averaged out via sensing}.

\subsubsection{Model I -- Shadowing-Free, Noise-Free}
This model assumes a \rev{constant but unknown \revv{path-loss-impaired} RSSI} at any location within the AoI. 
When there are no measurements, both users are equally valuable, since either can provide accurate estimation.
Once one is recruited, a second user has zero marginal value.
Hence, we have $v(\{1\})$ $=v(\{2\})$ $=v({\{1,2\}})=v_0$, where $v_0$ \rev{is the value (perceived by the platform) of estimating RSSIs at $U$}.
We set \revv{$v_0=4$} for later calculations.

\subsubsection{Model II (GP) -- Shadowing-Aware, Noise-Aware}
\revv{Under this model, $Z(x)$ for $x \in U$ is the sum of path-loss-impaired average RSSI $\mu(x)$ and spatially correlated shadowing $\delta(x) \sim N(0, \sigma^2_{x})$ with a covariance function $\mathcal{K}(d)$, where $d$ is the distance between two locations.
$Z(x_i)$ for $x_i \in S$ represents the noisy RSSI measured by user $i$, which includes additional hardware noise $\epsilon_i \sim (0, \sigma^2_{\epsilon_i})$ that is independent of shadowing.
Hence, RSSI at each location is modeled as a Gaussian random variable, and RSSIs at $S \cup U$ form a Gaussian random vector, whose joint distribution (Eq.~(\ref{eq:GP_RSSI_model_dist})) is defined by a mean vector and a covariance matrix (Eq.~(\ref{eq:covariance})).
More details about the model will be provided in Section~\ref{sec:GP_RSSI_model}.
Mutual information $MI(A)$ (Eq.~(\ref{eq:GP_RSSI_model_MI})) is used to quantify radio mapping quality of any subset of users $A\subseteq S$, and the valuation function $v(A)$ defined in Eq.~(\ref{eq:valuation_function}) (with $\alpha=0$) is used to translate $MI$ to value that is comparable to payments.}

We consider two cases under this model:
\begin{itemize}
\item Case 1: Users have different locations but provide data of the same quality. 
Suppose $x_1 = (-0.5,0)$, $x_2 = (0.5, 0.5)$ and $\sigma^2_{\epsilon_1}=\sigma^2_{\epsilon_2}=0.5$.

\item Case 2: Users have equally good (i.e., symmetric) locations but provide data of different quality.
Suppose $x_1 = (-0.5,0)$, $x_2 = (0.5, 0)$ and $\sigma^2_{\epsilon_1}=0.5 > \sigma^2_{\epsilon_2}=0.2$.

\end{itemize}

For illustration, we set $\mathcal{K}(d)=15.5 \cdot \exp(-\frac{d}{0.7})$ \cite{ying2015revisiting} (which is consistent with our later simulations in Section~\ref{sec:evaluation}) and set the currency $\kappa$ in $v(A)$ to $10$.
Valuation of users under both models is summarized in Table~\ref{table:example_sampling_value}.

The above highlights a few things. 
First, apart from the platform's preferences (e.g., choice of valuation function and parameters like $\kappa$), the RSSI data model plays an important role.
Second, for a reasonable RSSI data model that considers both shadowing and noise, user locations are important (e.g., Case $1$ of Model II); data quality also matters and affects valuation (e.g., Case $2$ of Model II). 

\begin{table}[]
\centering
\begin{tabular}{|c|c|c|c|c|}
\hline
 & $v(\emptyset)$ & $v(\{1\})$ & $v(\{2\})$ & $v(\{1,2\})$ \\ \hline
Model I & $0$ & $4$ & $4$ & $4$ \\ \hline
Model II - Case $1$ & $0$ & $2.18$ & $1.76$ & $3.48$ \\ \hline
Model II - Case $2$ & $0$ & $2.18$ & $2.23$ & $3.82$ \\ \hline
\end{tabular}
\caption{Valuation of users in the two-user example. }
\label{table:example_sampling_value}
\vspace{-0.5cm}
\end{table}

\subsection{How to Select and Pay Users}\label{sec:example_how_to_select_and_pay_users}
At this point, \rev{$v(\cdot)$ is available to the platform (Table~\ref{table:example_sampling_value}}). 
The next step is to select a subset of users and determine their \rev{price offers}.
This process is called ``incentivizing''.

\subsubsection{Deterministic Cost}\label{sec:deterministic_cost} 
If \rev{user devices are treated} as specialized spectrum sensors, then sensing costs are dominated by \textit{energy (or battery)} costs \cite{pottie2000wireless}.
In this case, it is reasonable to assume \rev{deterministic sensing costs, which can be inferred from the task and device type}.
Since the platform knows $\{c_i\}$, it can set $p_{\{i\}} = c_i$ to minimize payments while guaranteeing offer acceptance,  and search for the best set of users. 


Suppose \revv{$c_1=2$} and \revv{$c_2=1.5$}.
Then $p_{\{1\}}=2$, $p_{\{2\}}=1.5$ and $p_{\{1,2\}}=3.5$.
It is easy to see that $A^* = \{2\}$ leads to the maximum utility in each case.
\revv{If $c_1=c_2=1.5$, we have $p_{\{1\}}=p_{\{1\}}=1.5$ and $p_{\{1,2\}}=3$.
Under Model I, selecting either user will lead to a maximum utility of $2.5$ but not both.
In Case $1$ of Model II, the platform is better off with $A^*=\{1\}$, while it better selects both in Case $2$ of Model II.
Hence, selecting more users does not necessarily leads to a higher utility, since it also means a higher cost for the platform.}

\subsubsection{Random Cost with Known Distributions}\label{sec:example_known_cost_dist}
\rev{
In crowd-sensing, however, a user incurs an additional \textit{opportunity cost}, i.e., the loss of potential gain when the user decides to spend resources on sensing instead of other tasks.
It depends on the task and device status that varies over time.
Hence, the perceived sensing cost in each period consists of a deterministic energy cost and a random opportunity cost.

In this case, it makes sense to model user $i$'s sensing cost as a continuous random variable $C_i$ in $[\underline{c}_i, \bar{c}_i]$, where $\bar{c}_i \geq \underline{c}_i >0$ and $\underline{c}_i$ is the minimum energy cost.
Random variables $\{C_i\}$ are independent of each other.
Since device status is considered sensitive information, $C_i$ is private and only user $i$ knows its realization in each period, $c_i$, by evaluating the task and current device status\footnote{
\rev{In practice, we would expect a crowd-sensing application to be installed and running on users' mobile devices, which has some function that estimates the perceived sensing cost in each period based on the needed resources for the sensing task and the current device status.
Designing such a function for sensing cost estimation will be of practical importance, not only to this work, but also to many other crowd-sensing applications (e.g., \cite{koutsopoulos2013optimal}). 
But this topic is out of the scope of this paper, and will be left as future work.}
}.
We assume that the platform has only \textit{a priori} probabilistic knowledge of $C_i$.
Let $f_{C_i}(c_i)$ and $F_{C_i}(c_i)$ be the probability density function (PDF) and corresponding cumulative density function (CDF), respectively.
The PDF and $\underline{c}_i$, $\bar{c}_i$ could be learned by the platform from the empirical distribution out of prior cost declarations by users of the same device type, or from its long-term interaction with users (e.g., whether or not accept an offer with a known price).
If such prior information is absent, $C_i$ may be assumed to be uniformly distributed over $[\underline{c}_i$, $\bar{c}_i]$.
Hence, it is reasonable to assume that the platform can infer $f_{C_i}(c_i)$ or $F_{C_i}(c_i)$ based on the reported device type.

}

Since the platform does \textit{not} know $\{c_i\}$, it needs to consider possible offer rejections.
For $A=\{1\}$ or $\{2\}$, the uncertainty in user decisions implies the following utility,
\begin{equation}
u_{\{i\}} = 
\begin{cases}
    v_{\{i\}} - p_{\{i\}}, & \mbox{if user $i$ accepts the offer}\\
    0, &\mbox{otherwise}
\end{cases},
\end{equation}
which is a Bernoulli random variable and the acceptance probability is \rev{$\text{Pr}(c_i \leq p_{\{i\}}) = \int_{\underline{c}_i}^{p_{\{i\}}} f_{C_i}(c_i) dc_i = F_{C_i}(p_{\{i\}})$}.
In this case, it makes more sense to consider the EU, 
\begin{equation}
EU_{\{i\}} = \mathbb{E}[u_{\{i\}}] = (v_{\{i\}} - p_{\{i\}}) F_{C_i}(p_{\{i\}})
\label{eq:example_known_dist_max_utility_1},
\end{equation}
and the platform wants to find $p_{\{i\}}^*$ that maximizes $EU_{\{i\}}$, i.e., \textit{EU maximization}.
For $A=\{1,2\}$, the EU is given by
\begin{align}
EU_{\{1,2\}} =& (v_{\{1,2\}}-p_{\{1\}}-p_{\{2\}}) F_{C_1}(p_{\{1\}}) F_{C_2}(p_{\{2\}}) \label{eq:example_known_dist_max_utility_2} \\
&+ (v_{\{1\}}-p_{\{1\}}) F_{C_1}(p_{\{1\}}) (1-F_{C_2}(p_{\{2\}})) \nonumber \\
&+ (v_{\{2\}}-p_{\{2\}}) (1-F_{C_1}(p_{\{1\}}))F_{C_2}(p_{\{2\}}) \nonumber.
\end{align}
and the goal is to find ${\bf p}^*=[p_{\{1\}}^*,~p_{\{2\}}^*]$ that maximizes $EU_{\{1,2\}}$.
Note that in general the platform does not have to send out offers all at once and stop; it can send more backup based on the knowledge of outcomes of previous offers.

\revv{Suppose that} $C_1\sim \revv{U[1,2]}$ and $C_2\sim \revv{U[0.5, 1.5]}$, where $U[\cdot,\cdot]$ denotes the uniform distribution. 
For pricing, the platform's first thought could be setting $p_{\{i\}}=\bar{c}_i$.
Then following reasoning is the same with the deterministic-cost case. 
A natural generalization is to choose a desired probability of acceptance\footnote{Compared to choosing the same desired price for all users, it makes more sense to choose the same desired probability for all users, since users have different cost distributions in general. If all user devices are of the same type, then these two approaches are the same.} 
$\gamma \in [0,1]$ and set $p_{\{i\}}=F^{-1}_{C_i}(\gamma)$ for each user $i$, where $F^{-1}_{C_i}(\cdot)$ is the inverse CDF. 
\rev{Given $\gamma$, prices are fixed }and the platform wants to maximize the EU.

Taking Case $2$ of Model II and $\gamma=0.95$ as an example, we have $p_{\{1\}} = 1.95$, $p_{\{2\}}=1.45$, $p_{\{1,2\}}=3.4$, and $A^*=\{2\}$ is the best with $EU_{\{2\}}=(2.23-1.45)\cdot 0.95=0.74$ by Eq.~(\ref{eq:example_known_dist_max_utility_1}) and (\ref{eq:example_known_dist_max_utility_2}).
Note that the platform may further consider user 1, if user 2 rejects the offer. 
Then the overall EU with multi-batch offering would be $EU_{\{2\}}+(1-\gamma)EU_{\{1\}} > EU_{\{2\}}$.

Given $A$, $\gamma$ can also be optimized in each batch. 
Taking Case 2 of Model II as an example, when $A=\{1\}$, $EU_{\{1\}} = (2.18 - (1+\gamma))\gamma$ and $\gamma^* = 0.59$, $EU^*_{\{1\}}=0.35$.
Similarly, for $A=\{2\}$, $EU_{\{2\}}=(2.23-(0.5+\gamma))\gamma$ and $\gamma^*=0.865$, $EU^*_{\{2\}}=0.75$.
When $A=\{1,2\}$, we have $EU_{\{1,2\}}=-2.59 \gamma^2 + 2.91 \gamma$ and $\gamma^* = 0.56$, $EU_{\{1,2\}}^*=0.82$.
In this case, instead of using the same $\gamma$, the platform can also choose $\{\gamma_i\}$ for each user separately, and maximizing Eq.~(\ref{eq:example_known_dist_max_utility_2}) leads to $\gamma_1^*=0.37, \gamma_2^*=0.76$ and the resulting $EU^*_{\{1,2\}}=0.87$.

As we can see, the notation of utility accounts for  locations, data quality and sensing costs, and the notation of EU further considers possible offer rejections.
We can also see that user selection and \rev{price} determination are closely coupled in a pricing mechanism.
More will be discussed in Section~\ref{sec:pricing_mechanism}.


\section{\rev{Preliminaries and Our Model}}\label{sec:model}
In this section, we first provide background on submodularity. 
Then we present our spatial statistical model and define the metric for measuring radio mapping performance.
Finally, we present our valuation model and explore its properties.

\rev{
\subsection{Preliminaries}\label{sec:submodulartiy_background}
The submodularity property is formally defined as follows.

\begin{mydef}[Submodularity]\label{def:submodularity}
Let $\Omega$ be a finite set. 
A function $f: 2^\Omega \mapsto \mathbb{R}$ is submodular if for any $A,B \subseteq \Omega$, 
\begin{equation}
f(A) + f(B) \geq f(A\cup B) + f(A \cap B).
\end{equation}
\end{mydef}
Equivalently \cite{fujishige2005submodular}, a function $f$ is submodular if, for any $A \subseteq B \subseteq \Omega$ and any $i \in \Omega \setminus B$, 
\begin{equation}
f(A \cup \{i\}) - f(A) \geq f(B \cup \{i\}) - f(B).
\end{equation}
The notion of submodularity captures  \textit{diminishing returns} behaviors: adding a new element increases $f$ more, if there are fewer elements so far, and less, if there are more elements.

\begin{mydef}[(Approximately) monotonic function]
Let $\Omega$ be a finite set. 
A function $f: 2^\Omega \mapsto \mathbb{R}$ is said to be \textit{monotone} (or \textit{monotonic}), if $f(A \cup \{i\}) - f(A) \geq 0$ for any $A \subseteq \Omega$ and any $i \in \Omega\setminus A$; 
$f$ is said to be \textit{$\alpha$-approximately monotonic}, if $f(A \cup \{i\}) - f(A) \geq -\alpha$ for some small $\alpha > 0$, and for any $A\subseteq \Omega$ and any $i \in \Omega\setminus A$.
\end{mydef}

One of the most basic submodular maximization problems is USM, which is formally defined as follows.

\begin{mydef}[USM]
Given a nonnegative submodular fucntion $f:2^S \mapsto \mathbb{R}_+$,  $\max_{A \subseteq S} f(A)$ is called Unconstrained Submodular Maximization.
\end{mydef}

\begin{algorithm}
\caption{{\tt USM}}\label{algo:USM}
\LinesNumbered
\SetKwInOut{Input}{input}\SetKwInOut{Output}{output}

\Input{$S$ -- ground set, $f$ -- nonnegative submodular function } 
\Output{$A_n$ (or $B_n$) -- selected subset}
$A_0 \leftarrow \emptyset$, $B_0 \leftarrow S$\;
\ForEach{$i = 1$ to $n$ }{
    $a_i \leftarrow f(A_{i-1} \cup \{u_i\}) - f(A_{i-1})$\;
    $b_i \leftarrow f(B_{i-1} \setminus \{u_i\}) - f(B_{i-1})$\;
    \lIf{$a_i\geq b_i$}{
        $A_i \leftarrow A_i \cup \{u_i\}$, $B_i \leftarrow B_{i-1}$
    }
    \lElse{
       $A_i \leftarrow A_{i-1}$,  $B_i \leftarrow B_{i-1} \setminus \{u_i\}$
    }
}

\Return $A_n$ (or equivalently $B_n$)\;
\end{algorithm}

It is well known that USM is NP-hard \cite{fujishige2005submodular, feige2011Non-Monotone} and thus heuristic-based algorithms are often used to find approximate solutions.
One state-of-art \textit{linear-time} deterministic algorithm is proposed in \cite{buchbinder2015tight} and provided in Algorithm~\ref{algo:USM} for reference in the rest of this work.
It is essentially a greedy algorithm, and achieves a $1/3$-approximation, i.e., the algorithm obtains a solution $A$ with the guarantee that $f(A)\geq \frac{1}{3} f(OPT)$, where $OPT$ is the optimal solution.

}


\subsection{\rev{Spatial Statistical Model -- Gaussian Process (GP)}}\label{sec:GP_RSSI_model}
In this study, we employ GP \cite{deshpande2004model, krause2008near} (a generalization of Kriging) for radio mapping. 
Let the set of $n$ interested users be $S$, and the finely discretized AoI be $U$, where $|U| \gg |S|=n$, where $|\cdot|$ is the cardinality operator.
Define $V=S\cup U$ and each index $i \in V$ corresponds to a location $x_i$. 
\rev{Since the platform obtains \textit{noisy} RSSI measurements at $S$ and wants to estimate \textit{noiseless} front-end RSSIs at $U$, the RSSI $Z(x_i)$ or $Z_i$ is modeled as a Gaussian random variable in GP, 
}
\begin{equation}
Z(x_i) = 
\begin{cases}
\mu(x_i) + \delta(x_i),  &\mbox{for } i\in U\\
\mu(x_i) + \delta(x_i) + \epsilon_i, &\mbox{for } i \in S
\end{cases}, ~~~~~(dBm)\label{eq:GP_RSSI_model}
\end{equation}
where $\mu(x_i)$ is path-loss-impaired RSSI, $\delta(x_i)\sim N(0, \sigma^2_{x_i})$ is spatially correlated shadowing and $\epsilon_i\sim N(0, \sigma^2_{\epsilon_i})$ is hardware noise of user $i$'s device.

Define a kernel (or covariance) function $\mathcal{K}(\cdot, \cdot)$  such that $\mathcal{K}(i, j)$ is the covariance between $\delta(x_i)$ and $\delta(x_j)$.
In GP, the RSSIs at $V$ form a Gaussian random vector $Z_V=[Z(x_i)]_{i\in V}$ with a joint distribution of
\begin{equation}
f_{Z_V}(z_{V}) = \frac{1}{(2\pi)^{n/2} |\Sigma_{VV}|} e^{-\frac{1}{2} (z_{V} - \mu_{V})^T \Sigma^{-1}_{V V} (z_{V} - \mu_{V}) },\label{eq:GP_RSSI_model_dist}
\end{equation}
where $z_V=[z(x_i)]_{i\in V}$ is a realization of $Z_V$, $\mu_V=[\mu(x_i)]_{i \in V}$ is the mean vector and $\Sigma_{V V}$ is the covariance matrix.
For any pair of indices $i,j \in V$, their covariance $\sigma_{ij}$ is the $(i,j)$-th entry of $\Sigma_{VV}$, which is given by
\begin{equation}
\sigma_{ij} = 
\begin{cases}
\mathcal{K}(i,j), &\mbox{if } i \neq j \\
\mathcal{K}(i,j) \mbox{ or } \sigma^2_{x_i}, &\mbox{if } i=j \in U\\
\mathcal{K}(i,j) + \sigma^2_{\epsilon_i} \mbox{ or } \sigma^2_{x_i} + \sigma^2_{\epsilon_i}, &\mbox{if } i=j \in S
\end{cases}
\label{eq:covariance}
\end{equation}

Given a set of measurements $Z_A$ where $A\subseteq S$, $Z(x_i)$ is a conditional Gaussian random variable with a mean $\mu_{Z(x_i)|Z_A}$ (or simply $\mu_{i|A}$) and a variance of $\sigma^2_{Z(x_i)|Z_A}$ (or simply $\sigma^2_{i|A}$),
\begin{align}  \label{eq:cond_mean}
\mu_{i|A} &= \mu(x_i) + \Sigma_{Ai}^T \Sigma^{-1}_{AA} (z_{A} - \mu_{A}), \\ 
\sigma^2_{i | A} &= \sigma_{ii} - \Sigma_{Ai}^T \Sigma^{-1}_{ A A}\Sigma_{A i}. \label{eq:cond_var}
\end{align}
Note that the posterior variance in Eq.~(\ref{eq:cond_var}) only depends on $\Sigma_{V V}$, not the actual measured values $z_A$. 

Estimating $\mathcal{K}(\cdot, \cdot)$ can be difficult in practice, and it is often assumed that $\mathcal{K}(\cdot, \cdot)$ is stationary (i.e., a function of location displacement) and isotropic (i.e., a function of distance).
In other words, $\mathcal{K}(i, j) = \mathcal{K}_\theta(||x_i - x_j||)$, where $\theta$ is a set of parameters. 
That being said, our following discussions do not assume stationarity or isotropy, and thus can be applied to general kernel functions.
But we do assume both mean and kernel functions have been estimated from previous measurements\footnote{In \cite{phillips2012practical}, authors used a predictive (empirical) path loss model to estimate the mean process $\mu(x)$. This procedure is called \textit{detrending}. In the same paper as well as \cite{ying2015revisiting}, authors estimated an empirical semivarigram $\gamma(\cdot)$ (isotropic and stationary) from real measurements and fitted it with parametric models. The relationship between $\gamma(\cdot)$ and $\mathcal{K}(\cdot, \cdot)$ is $\mathcal{K}(i, j)=c_0-\gamma(||x_i-x_j||)$ for $i \neq j$, where $c_0$ is some constant. } and available in the current period. 

\subsection{\rev{Mutual Information (MI) for Uncertainty Reduction}}\label{sec:sampling_design_and_submodularity}
To measure radio mapping performance, we adopt the MI metric \cite{krause2008near}, which is defined as follows,
\begin{equation}
MI(A) = I(Z_A; Z_{V \setminus A}) = H(Z_{V\setminus A}) - H(Z_{V\setminus A}|Z_A),\label{eq:GP_RSSI_model_MI}
\end{equation}
which is the amount of uncertainty reduction about RSSIs at unmeasured locations given measurements at $A$. 

\rev{Note that the platform is interested in $Z_{V\setminus A}$, which includes RSSIs at $S\setminus A$ (i.e., locations with confirmed user presence) and $U$ (i.e., locations with possible user presence)}.
As implicitly assumed in \cite{krause2008near}, $Z(x_i)$ includes noise for $i\in S\setminus A$ in the definition of MI, which is not a big issue, since noise is relatively small.
Compared to the entropy criterion $H(Z_{V\setminus A} | Z_A)$, MI tends to not select users along the boundaries and avoids the ``waste'' of information.

Denote by $MI(i|A)$ the marginal MI of an additional user $i\in S\setminus A$ given $A$.
It is given by
\begin{align}\label{eq:comp_marginal_MI}
MI(i|A) &= MI(\{i\}\cup A) - MI(A)\\
&= H(Z_i| Z_A) - H(Z_i| Z_{V\setminus(A \cup \{i\})}),
\end{align}
where $H(Z_i| Z_A)=\frac{1}{2}\log(2\pi e \sigma^2_{i| A})$ is the conditional entropy, and it can be easily computed from Eq.~(\ref{eq:cond_var}). 

It has been shown in \cite{krause2008near} that $MI(A)$ is both  submodular and $\alpha$-approximately monotone\footnote{\rev{Intuitively, $MI(A)$ is monotone under the condition that $|V|\gg |S|\geq |A|$ and thus adding one more user increases the MI.
Otherwise, consider the example that $S=\{1,2\}$ and $|V|=|S|$, then $MI(\emptyset)=MI(\{1,2\})=0$ but $MI(\{1\})>0$ and $MI(\{2\})>0$, which means that MI first increases then decreases as more users are selected.
Its monotonicity is approximate due to the extreme case where there exist two (or more) users arbitrarily close to each other.
If one is selected, selecting the other one will decrease MI. 
More discussions are available in \cite{krause2008near}.
In practice, the platform can avoid such extreme cases by considering only one of them.
Also, a pricing mechanism that maximizes the (expected) utility should not select both, since the second user is not beneficial for radio mapping and not free-of-charge.} \label{footnote:approximate_monotone}
}.
For any $\alpha>0$, a discretization level exists so that $MI(A)$ is approximately monotone.

\subsection{Valuation Function}\label{sec:valuation_model}
We consider the following valuation function $v: 2^S \mapsto \mathbb{R}_+$ for the platform,
\begin{equation}
v(A) = \kappa \cdot \log (1 + MI'(A)),\label{eq:valuation_function}
\end{equation}
where $\kappa> 0$ is a constant and $MI'(A)=MI(A) + \alpha |A|$.
\rev{Intuitively, $\kappa$ is the currency that reflects the platform's preference over per unit  MI (in log scale).
Commonly used in economics, $\log(\cdot)$ further emphasizes the diminishing returns behavior.
We introduce $\alpha |A|$ to ignore the extreme case where some users are arbitrarily close to each other, which rarely occurs and/or can be avoided in practice (see Footnote \ref{footnote:approximate_monotone}). 
}

We show that there exists useful structural properties like submodularity and monotonicity in $v(A)$.

\begin{mylemma}\label{lemma:v_submodular}
The valuation function $v(\cdot)$ in Eq.~(\ref{eq:valuation_function}) is monotone submodular.
\end{mylemma}
\begin{proof}
See Appendix~\ref{proof:v_submodular} for proof.
\end{proof}


\section{Pricing Mechanism}\label{sec:pricing_mechanism}
\rev{In this section, we formulate pricing mechanism design as expected utility (EU) maximization and propose two schemes: (1) sequential offering and (2) batched offering.}

\subsection{EU Maximization}\label{sec:EU_maximization}
Given $S$, $v(\cdot)$ and $\{F_{C_i}(c_i)\}$, the platform wants to determine a set of offers $(A, {\bf p})$, where $A \subseteq S$ are selected users and ${\bf p}=[p_i]_{i \in A}$ is the corresponding price vector.
Let the decision of the $i$-the selected user be $X_i$, which is given by 
\begin{equation}
X_i=
\begin{cases}
1, & \mbox{if } c_i \leq p_i \mbox{ (i.e., offer is accepted)} \\
0, & \mbox{else (i.e., offer is rejected)}
\end{cases}. \label{eq:RV_user_decision}
\end{equation}
It is a Bernoulli random variable (from the platform's perspective), and $\text{Pr}(X_i=1)=\int_{\underline{c}_i}^{p_i} f_{C_i}(c_i) dc_i = F_{C_i}(p_i)$. 

\rev{
As mentioned in Section~\ref{sec:sys_arch}, an offer may be expired, and this event is modeled by a random variable $X_i'$, i.e.,
\begin{equation}
X_i' = 
\begin{cases}
1, & \mbox{if offer is unexpired}\\
0, & \mbox{if offer is expired}
\end{cases},
\end{equation}
where $\rho_i = \text{Pr}(X_i'=1)$ is the probability of an unexpired offer.
We assume that the platform can estimate $\rho_i$ and that $X'_i$ is independent of $X_i$.
}

\rev{
Let $Y_i$ be a random variable that represents whether a user is successfully \textit{recruited} (i.e., offer is unexpired and accepted),
\begin{equation}
Y_i = 
\begin{cases}
1, & \mbox{if offer is unexpired AND accepted}\\
0, & \mbox{if offer is expired OR rejected}
\end{cases},
\end{equation}
where 
\begin{align}
\gamma_i &= \text{Pr}(Y_i = 1) = \text{Pr}(X_i'=1, X_i=1) \nonumber \\
&= \text{Pr}(X_i'=1) \cdot \text{Pr}(X_i = 1 | X_i' = 1) \nonumber \\
&= \rho_i \cdot F_{C_i}(p_i) \in [0, \rho_i],\label{eq:prob_of_recruited}
\end{align}
is the probability that the $i$-th selected user is recruited.
}

Define ${\bf Y} = [Y_i]_{i \in A}$ and let ${\bf y}$ be the realization of ${\bf Y}$.
Then $A_{\bf y} \subseteq A$ is the set of recruited users.
Then the EU is given by
\begin{align}
EU(A, {\bf p}) = \mathbb{E}_{\bf Y} [u(A_{\bf y}, {\bf p})]  = \sum_{{\bf y}} \text{Pr}(A_{\bf y}, {\bf p}) u(A_{\bf y}, {\bf p}),\label{eq:EU_def}
\end{align}
where 
\begin{align}
\text{Pr}(A_{\bf y}, {\bf p})&=\prod_{i\in A_{\bf y}} \gamma_i \cdot \prod_{i\notin A_{\bf y}} (1- \gamma_i), \\
u(A_{\bf y}, {\bf p}) &= v(A_{\bf y}) - \sum_{i\in A_{\bf y}} p_i,
\end{align}
\normalsize
are the probability and utility of $A_{\bf y}$ given ${\bf p}$, respectively.

The goal of the platform is to design a \textit{pricing mechanism} based on \textit{EU maximization}, that is,
\begin{equation}
\max_{A\subseteq S, {\bf p}} EU(A, {\bf p}) \label{eq:EU_maximization}.
\end{equation}
In this sense, a pricing mechanism consists of a \textit{selection rule} and a \textit{pricing rule}, which is joint optimization in the discrete domain of $A$ and the continuous domain of ${\bf p}$.

\subsection{Sequential Offering}\label{sec:sequential_offering}
\rev{We first consider a special case of EU maximization, where $|A|=1$.
That is, the platform only selects one best user with its best offer in each round, and waits for its decision before making the next offer.}
We call it \textit{sequential (individual) offering}.
Formally, the task in each round is 
\begin{align}
\max_{i \in S\setminus A, p_i} EU(A \cup \{i\}, [{\bf p},~p_i] | {\bf Y} = {\bf y}) 
\label{eq:sequential_offering},
\end{align}
where ${\bf y}$ represents the outcomes of offers that have been sent so far and is known to the platform.


\begin{algorithm}
\caption{{\tt Sequential\_Offering}}
\label{algo:sequential_offering}
\LinesNumbered
\SetKwInOut{Input}{input}\SetKwInOut{Output}{output}

\Input{$S$ -- set of users, $v(\cdot)$ -- valuation function, $\{F_{C_i}(\cdot)\}$ -- cost distributions, $\{\rho_i\}$ -- probabilities of  unexpired offers, $\tau$ -- threshold
} 
\Output{$A$ -- selected users, ${\bf p}$ -- prices, ${\bf y}$ -- outcomes}
$A\leftarrow \emptyset$, ${\bf p} \leftarrow NULL$, ${\bf y}\leftarrow NULL$\;
\While{$A \neq S$}{
    \ForEach{each user $i$ in $S\setminus A$}{
        $p_i^* \leftarrow \arg\max_{p_i\in [\underline{c}_i, \bar{c}_i]} [v(i|A_{\bf y}) - p_i] \cdot F_{C_i}(p_i)$\;
        $EU_i \leftarrow [v(i|A_{\bf y}) - p^*_i] \cdot \rho_i \cdot F_{C_i}(p^*_i)$
    }
    $i^* \leftarrow \arg \max_{i\in S \setminus A} EU_i$\;
    \While{$EU_{i^*}> \tau$}{
        Send the offer $(i^*, p_{i^*}^*)$ and observe $y_{i^*}$\;
        $A\leftarrow A\cup\{i^*\}$, ${\bf p}\leftarrow[{\bf p}, p_{i^*}^*]$, ${\bf y}\leftarrow [{\bf y}, y_{i^*}]$\;
        \lIf{$y_{i^*}=1$}{\textbf{break}}
        \lElse{
            $i^* \leftarrow \arg \max_{i\in S \setminus A} EU_i$ 
        }
    }
}

\Return $A$, ${\bf p}$, ${\bf y}$\;
\end{algorithm}

The algorithm for sequential offering is described in Algorithm~\ref{algo:sequential_offering}.
\revv{The idea is as follows: The platform first determines an optimum price $p_i^*$ tailored to each $i$ that maximizes $EU_i$. 
Then it picks the index $i^*$ that maximizes among the $EU_i$, and offers to that user the corresponding $p_{i^*}^*$.}

\subsubsection{Price Determination (Lines 3-5)}
Depending on whether user $i \in S\setminus A$ is successfully recruited, the utility is
\begin{align}
&~~~~u(A_{\bf y} \cup \{i\}_{Y_i}, [{\bf p},~p_i]) \nonumber \\ 
&= u(A_{\bf y}, {\bf p}) + 
\begin{cases}
v(i|A_{\bf y})-p_i, &\mbox{if } Y_i=1 \\
0, & \mbox{otherwise},
\end{cases}
\end{align}
\normalsize
where $v(i|A_{\bf y})=v(\{i\}\cup A_{\bf y}) - v(A_{\bf y})$ is the marginal value of $i$ given $A_{\bf y}$.
The task is to find
\begin{align}
p^*_i 
&=\arg \max_{p_i \in [\underline{c}_i, \bar{c}_i] } \mathbb{E}_{Y_i}[u(A_{\bf y} \cup \{i\}_{Y_i}, [{\bf p},~p_i])] \nonumber \\
&= \arg \max_{p_i \in [\underline{c}_i, \bar{c}_i] } [v(i|A_{\bf y}) - p_i]\cdot \text{Pr}(Y_i=1) \nonumber \\
&= \arg \max_{p_i \in [\underline{c}_i, \bar{c}_i] } [v(i|A_{\bf y}) - p_i] \cdot F_{C_i}(p_i) \label{eq:sequential_offering_best_offer}.
\end{align}
Note that $\rho_i$ in $\text{Pr}(Y_i=1)$ in Eq.~(\ref{eq:prob_of_recruited}) \revv{is a constant and does not impact the choice of} $p_i^*$.
If $f_{C_i}(c_i) = F'_{C_i}(c_i)$ is differentiable and non-increasing,  the objective function in Eq.~(\ref{eq:sequential_offering_best_offer}) will be concave in $p_i$, and $p_i^*$ can be obtained with efficient algorithms (e.g., gradient descent).
\rev{If $F_{C_i}(c_i)$ is twice continuously differentiable, techniques like interval analysis may be used to find $p^*$ \cite{hansen1979global}. 
}

\subsubsection{User Selection (Line 6)}
The best user $i^*$ that maximizes the EU is found, 
\begin{equation}
i^* = \arg \max_{i \in S\setminus A} [v(i|A_{\bf y}) - p^*_i]\cdot \rho_i \cdot F_{C_i}(p^*_i).
\end{equation}
Note that the above selection also takes $\rho_i$ into account.
If the user is recruited (Lines 10), the algorithm will go to Line 3 to recompute best prices for remaining users; otherwise, it sends out the next best offer immediately until one is accepted (Lines 7-11). 
To enable fast convergence, the platform can set a minimum threshold $\tau>0$ (e.g., 0.01) for the marginal EU (Line 7). 
The platform stops making offers when there are (1) no remaining users or (2) none of the remaining users leads to a non-trivial marginal EU.

\subsubsection{Complexity Analysis}
If we assume $\mathcal{O}(1)$ for computing the best price for a single user, the overall computational complexity of Algorithm~\ref{algo:sequential_offering} is $\mathcal{O}(n^2)$, since it takes $\mathcal{O}(n)$ to compute best prices for all remaining users and may select up to $n$ users in the worst case.
The inner while-loop does not require re-computation of best prices and is dominated by the for-loop. 
Note that $\mathcal{O}(n^2)$ is very conservative, since the algorithm may stop much earlier based on the configuration.

\subsection{Batched Offering}
\rev{
As we can see, sequential offering is intuitive and straightforward, but its main drawback is the (possibly) large delay accumulated over multiple rounds of offering.
Hence, a natural generalization is to make multiple offers (i.e., a batch) in each round and continue offering for multiple rounds.}
We refer to it as \textit{(sequential) batched offering}.

In batched offering, the platform is faced with the general case of EU maximization in Eq.~(\ref{eq:EU_maximization}) in each round. 
\rev{
    Unfortunately, joint optimization can be difficult in practice, mainly because $EU(A, {\bf p})$ is a multi-variate function in the continuous domain of ${\bf p}$ given $A$, and there may not exist structural properties like concavity in general to enable efficient computation of the global optimum. 
    Exhaustive search is prohibitive as the space of ${\bf p}$ is huge.
    Fortunately, $EU(A, {\bf p})$ has a useful structural property (i.e., submodularity) in the discrete domain of $A$ as a set function, which inspires our following pricing mechanism design.
}

\begin{mylemma}\label{lemma:EU_submodular}
Given ${\bf p}$, $EU(A, {\bf p})$ is submodular in $A$.
\end{mylemma}
\begin{proof}
See Appendix~\ref{proof:EU_submodular} for proof.
\end{proof}

The basic idea of our pricing mechanism is to first fix the pricing rule and then focus on user selection to exploit the submodularity property.
As mentioned in Section~\ref{sec:submodulartiy_background}, if a set function $f$ is nonnegative submodular and the problem is $\max_{A\subseteq S} f(A)$, there exist heuristic-based algorithms (e.g., Algorithm~\ref{algo:USM}) that provide  solutions with performance guarantee at low complexity.
Next, we will present our pricing mechanisms for single-batch and multi-batch offering.



\subsubsection{Single-Batch Offering}
As mentioned in Section~\ref{sec:example}, we consider the following pricing rule in this work: the platform chooses a desired probability of recruitment $\gamma \in (0,1]$ such that \rev{$\gamma_i = \min(\gamma, \rho_i)$} for any $i \in S$ and determines corresponding prices, i.e., 
\begin{equation}
p_{\gamma}(A)=\sum_{i\in A} p_{\gamma}(\{i\}) = \sum_{i\in A} F_{C_i}^{-1}\left( \min (\gamma/\rho_i, 1 ) \right). \label{eq:pricing_rule}
\end{equation}
Given $p_{\gamma}(\cdot)$, user selection then becomes
\begin{equation}\label{eq:EU_gamma_A}
\max_{A\subseteq S} EU_{\gamma}(A) =  \max_{A\subseteq S} \sum_{{\bf y}} \text{Pr}_{\gamma}(A_{\bf y}) u_{\gamma}(A_{\bf y}),
\end{equation}
where 
$\text{Pr}_{\gamma}(A_{\bf y}) = \prod_{i \in A_{\bf y}} \gamma_i \cdot \prod_{i \notin A_{\bf y}} (1-\gamma_i)$ and $u_{\gamma}(A_{\bf y}) = v(A_{\bf y}) - p_{\gamma}(A_{\bf y})$.
By Lemma~\ref{lemma:EU_submodular}, $EU_{\gamma}(A)$ is submodular.

However, the USM formulation also requires the objective function to be nonnegative, but $EU_{\gamma}(A)$ can be negative. 
To bypass this issue, one straightforward way is to define 
\begin{equation}
EU'_{\gamma}(A)=EU_{\gamma}(A) + p_0
\end{equation}
where 
\begin{equation}
p_0 = \sum_{i \in S} \gamma_i p_\gamma(\{i\})
\end{equation} 
is a constant that represents the maximum expected price, and adding a constant preserves submodularity.
It is easy to see that $EU'_{\gamma}(A)$ is both submodular and nonnegative, and $\max_{A \subseteq S} EU_\gamma'(A)$ is equivalent to $\max_{A \subseteq S} EU_\gamma'(A)$.

\rev{Another issue is due to the difficulty of analytically evaluating $EU_{\gamma}(A)$ (or $EU'_\gamma(A)$), since it involves an exponentially growing number of terms due to the summation in Eq.~(\ref{eq:EU_gamma_A}).
In practice, the Monte-Carlo (MC) method \cite{rubinstein2016simulation} can be used to obtain estimates of $\hat{EU_\gamma}(A)$ as well as $\hat{EU_\gamma'}(A)$ (by adding the constant $p_0$ to $\hat{EU_\gamma}(A)$).
We show that {\tt USM}($S$, $\hat{EU_\gamma'}$) has the following performance.}

\rev{\begin{mytheorem}\label{theorem:USM_with_approx_EU_bound}
If $|\hat{EU_\gamma'}(A)-EU_\gamma'(A)| \leq \epsilon$ for some small $\epsilon>0$ for any $A\subseteq S$, {\tt USM}($S$, $\hat{EU_\gamma'}$) (or equivalently {\tt USM}($S$, $\hat{EU_\gamma}$)) returns a solution $A$ with the following performance,
\begin{align}
EU_\gamma'(A) & \geq \frac{1}{3} EU_\gamma'(OPT) - \frac{1}{3}(2n + 2) \epsilon,  \text{ and } \\
EU_\gamma(A) & \geq \frac{1}{3} EU_\gamma(OPT) - \frac{2}{3} p_0 - \frac{1}{3}(2n + 2) \epsilon, 
\end{align}
where $OPT$ is the optimal solution for $\max_{A \subseteq S} EU_\gamma'(A)$ as well as $\max_{A \subseteq S} EU_\gamma(A)$, $p_0 = \sum_{i \in S} \gamma_i p_\gamma(\{i\})$ and $n = |S|$.
\end{mytheorem}
\begin{proof}
See Appendix~\ref{proof:USM_with_approx_EU_bound} for proof.
\end{proof}}

\begin{algorithm}
\caption{{\tt Single\_Batch\_Offering}}
\label{algo:single_batch_offering}
\LinesNumbered
\SetKwInOut{Input}{input}\SetKwInOut{Output}{output}

\Input{$S$ -- set of users, $f$ -- $\hat{EU'_{\gamma}}(\cdot)$ (or equivalently $\hat{EU}_\gamma(\cdot)$), 
$\Gamma=[\gamma_1,...,\gamma_l]$ -- candidate $\gamma$ values where $\gamma_1 < \gamma_2 <...<\gamma_l$, $p_\gamma(\cdot)$ -- pricing rule
} 
\Output{$A$ -- selected users, ${\bf p}$ -- prices, $\gamma^*$ -- best probability of recruitment}

$A \leftarrow \emptyset$, ${\bf p}\leftarrow NULL$, $\gamma^*\leftarrow 0$\;
\ForEach{$\gamma$ in $[\gamma_1,...,\gamma_l]$}{
	$A_\gamma \leftarrow$ {\tt USM}($S, f$)\; 
	\lIf{$A_\gamma = \emptyset$}{{\bf break}}
	\If{$\hat{EU}_{\gamma}(A_\gamma) > \hat{EU}_{\gamma}(A)$}{
	$A \leftarrow A_\gamma$, $\gamma^* \leftarrow \gamma$\;
	}
}

\ForEach{each user $i$ in $A$}{${\bf p} \leftarrow [{\bf p}, p_{\gamma^*}(\{i\})]$}

\Return $A$, ${\bf p}$, $\gamma^*$\;
\end{algorithm}

Algorithm~\ref{algo:single_batch_offering} describes the algorithm for single-batch offering. 
Since each $\gamma$ leads to a different solution $A_\gamma$, it is better for the platform to search through a list of $l$ candidate $\gamma$ values (e.g. $\{0.1, 0.2, ..., 1.0\}$) to find the best one that maximizes $\hat{EU}_{\gamma} (A_\gamma)$.
As {\tt USM} takes $\mathcal{O}(n)$ time, the overall complexity of Algorithm~\ref{algo:single_batch_offering} is $\mathcal{O}(ln)$.
%
%

\subsubsection{Multi-Batch Offering}
In the case of expired or rejected offers in the previous batch, the platform may send out more batches until the next batch is no longer profitable.

Denote by $EU_{\gamma}(B|A_{\bf y})$ the marginal EU of additional offers $B$ conditioned on the set of recruited users $A_{\bf y}$, 
\begin{align}
EU_{\gamma}(B | A_{\bf y})&=\mathbb{E}_{\bf Y'} [u_{\gamma}(B_{\bf y'}|A_{\bf y})],
\end{align}
where $u_{\gamma}(B_{\bf y'}|A_{\bf y}) = u_{\gamma}(B_{\bf y'}\cup A_{\bf y}) - u_{\gamma}(A_{\bf y}) = v(B_{\bf y'}|A_{\bf y})-p_{\gamma}(B_{\bf y'})$ is the marginal utility of $B_{\bf y'}$ given $A_{\bf y}$.
We can see from Lemma~\ref{lemma:EU_submodular}  that $EU_{\gamma}(B|A_{\bf y})$ is again submodular in $B$.

\begin{algorithm}
\caption{{\tt Multi\_Batch\_Offering}}
\label{algo:multi_batch_offering}
\LinesNumbered
\SetKwInOut{Input}{input}\SetKwInOut{Output}{output}

\Input{$S$ -- set of users, $f$ -- $\hat{EU'}_{\gamma}(\cdot)$ (or equivalently $\hat{EU}_\gamma(\cdot)$), $\Gamma=[\gamma_1,...,\gamma_l]$ -- candidate $\gamma$ values where $\gamma_1 < \gamma_2 <...<\gamma_l$, $p_\gamma(\cdot)$ -- pricing rule, $\tau$ -- threshold} 
\Output{$A$ -- selected users, ${\bf p}$ -- prices, ${\bf y}$ -- outcomes}
$A \leftarrow \emptyset$, ${\bf p} \leftarrow NULL$, ${\bf y}\leftarrow NULL$ \;
\While{$A \neq S$}{
    $(B, {\bf p}_B, \gamma^*)$ $\leftarrow${\tt Single\_Batch\_Offering}($S \setminus A, f(\cdot|A_{\bf y}), \Gamma, p_\gamma(\cdot)$) \;
    \If{$\hat{EU}_{\gamma^*}(B|A_{\bf y}) > \tau$}{
        Send out offers $(B, {\bf p}_B)$ and observe  ${\bf y}_B$\;
        $A \leftarrow A\cup B$, ${\bf p}\leftarrow [{\bf p}, {\bf p}_B]$, ${\bf y}\leftarrow[{\bf y}, {\bf y}_B]$
    }
    \lElse{ \textbf{break}}
}

\Return $\mathcal{A}$, ${\bf p}$, ${\bf y}$\;
\end{algorithm}

The algorithm for multi-batch offering is provided in Algorithm~\ref{algo:multi_batch_offering}.
Starting from the second batch, the marginal EU function is passed as an input to  {\tt Single\_Batch\_Offering} (Line 3).
If the (estimated) marginal EU of $B$ is higher than a preset threshold $\tau>0$ (e.g., 0.01), then it is profitable on average to send out the next batch of offers (Lines 5).
The platform will then wait for the results and update $A$, ${\bf p}$, ${\bf y}$ accordingly (Line 6).
The offering process stops when (1) there are no more users to consider (Line 2), or (2) the next batch is no longer profitable on average (Line 4). 

\section{Evaluation}\label{sec:evaluation}
In this section, we conduct simulations to evaluate proposed pricing mechanisms based on EU maximization, and compare them against baseline mechanisms based on (best-case) utility maximization.
We also study the impact of offer expiration on the proposed pricing mechanisms. 

\subsection{Simulation Setup}
Fig.~\ref{fig:60_users_poisson_topo} is a sample topology of 60 users, whose locations are randomly generated from the spatial Poisson process. 
The AoI is a 6km-by-6km region,
discretized into a total of 169 points with a resolution of 450 meters. 
We assume an exponential kernel function $\mathcal{K}(d) = 15.5\cdot \exp(-\frac{d}{0.7})$, as shown in Fig.~\ref{fig:cov_func}, which is adapted from the semivariogram fitted from real measurements \cite{ying2015revisiting}. 
Given the above settings, negative marginal MI values are not observed and $\alpha$ is set to $0$.

The domain of $C_i$ is $[\underline{c}_i, \underline{c}_i + \Delta c]$, where $\Delta c > 0$ and $\underline{c}_i$ is randomly generated from $U[0.1, 0.2]$, where $U[\cdot, \cdot]$ denotes the uniform distribution.
We consider two types of distributions: uniform (UN) and truncated normal (TN) distributions. 
TN is the normal distribution $N(\underline{c}_i, (\Delta c/3)^2)$ truncated to $[\underline{c}_i, \underline{c}_i + \Delta c]$.
Compared to UN, TN represents the situation where the majority of users have sensing costs closer to the energy costs $\underline{c}_i$, despite of opportunity costs.
The same set of noise variances independently drawn from $U[0.5,1]$ is used throughout our simulation.

\begin{figure}[ht!]
        \centering
        \begin{subfigure}[b]{0.45\columnwidth}
                \includegraphics[width=1\textwidth]{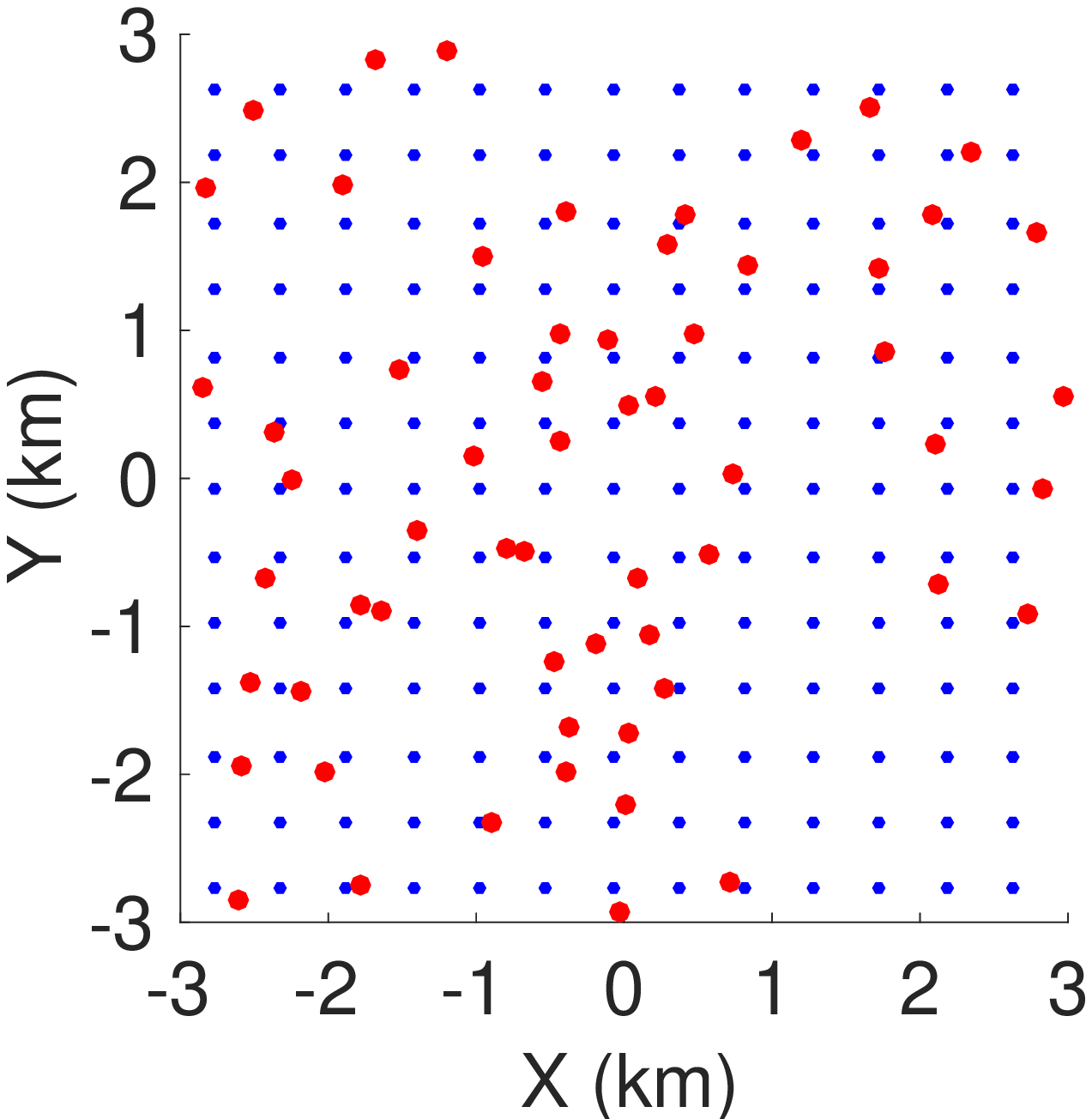}
                \caption{Sample topology}
                \label{fig:60_users_poisson_topo}
        \end{subfigure}
        \begin{subfigure}[b]{0.45\columnwidth}
                \includegraphics[width=1\textwidth]{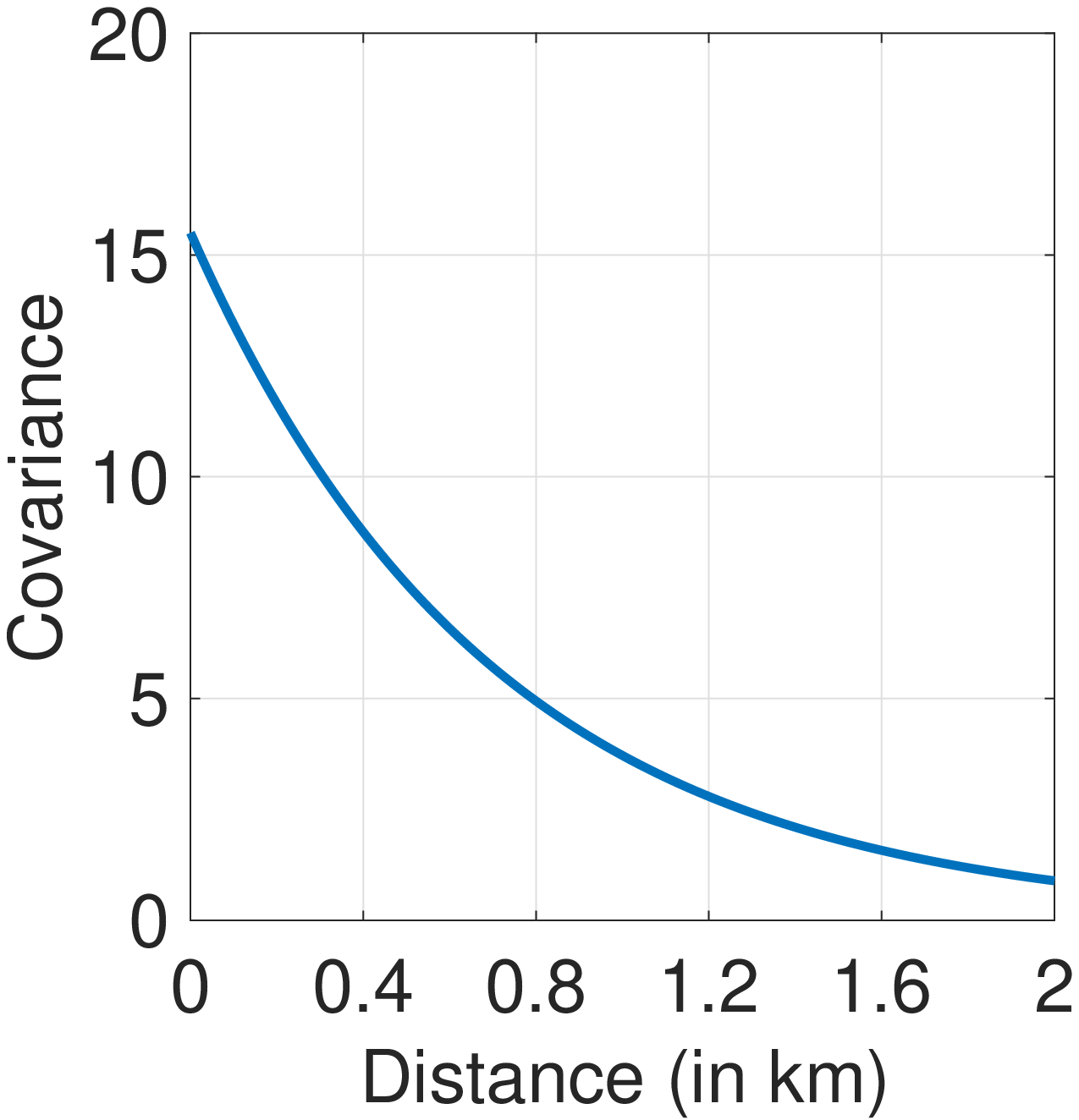}
                \caption{Kernel function}
                \label{fig:cov_func}
        \end{subfigure}
        \caption{(a) Sample topology of 60 users (in red) in a 6km-by-6km area that is discretized into a mesh grid of 169 points (in blue). (b) The kernel function $\mathcal{K}(d) = 15.5\cdot \exp(-\frac{d}{0.7})$.}
        \label{fig:simu_topology_and_cov_func}
        \vspace{-0.3cm}
\end{figure}

\subsubsection{Baseline Mechanisms}
As mentioned in Section~\ref{sec:example}, an alternative  to design a pricing mechanism is to maximize the best-case utility as in Eq.~(\ref{eq:goal_max_utility}), assuming no expired or rejected offers.
In this simulation, we consider single-batch and multi-batch offering based on utility maximization as baseline mechanisms.
That is, instead of passing $\hat{EU_\gamma'}(\cdot)$ into {\tt USM} in each batch (Line 3 of Algorithm~\ref{algo:single_batch_offering}), we pass $u'_\gamma(\cdot)$ (a nonnegative submodular function) into {\tt USM}, where $u'_\gamma(A)=u_\gamma(A)+p_\gamma(S)$ and $u_\gamma(A) = v(A)-p_\gamma(A)$.
Compared to $EU'_\gamma(\cdot)$, the objective function $u'_\gamma(\cdot)$ does not require the MC method and is easy to evaluate. 

For convenience, we refer to EU-maximization-based mechanisms by feeding $\hat{EU'_\gamma}(\cdot)$ to {\tt USM} as USM-EU, and baseline mechanisms by feeding $u'_\gamma(\cdot)$ to {\tt USM} as USM-u.

\subsection{USM-EU vs. USM-u in Single-Batch Offering}
In this experiment, we compare the performance of USM-EU and USM-u in singe-batch offering.
$\hat{EU_\gamma'}(\cdot)$ is obtained by averaging over $50$ iterations of MC simulations.
We randomly select $30$ or $60$ users, and set $\kappa$ in $v(\cdot)$ (Eq.~(\ref{eq:valuation_function})) to $4$ or $8$, and $\Delta c$ to $0.1$ or $0.5$.  
\rev{We assume no expired offers and set $\rho_i=1$ for each user $i$.
We will study the impact of offer expiration later in Section~\ref{sec:impact_of_rho}.} 
A total of $30$ iterations are conducted, and a different seed is used for generating users and cost distributions in each iteration.
In the $i$-th iteration, however, the same set of users and cost distributions are used across different $\gamma$ and mechanisms for fair comparison.
Results are provided in Fig.~\ref{fig:USM_u_vs_USM_EU_avg}.

\begin{figure}[h!]
        \centering
        \begin{subfigure}[b]{0.45\columnwidth}
                \includegraphics[width=1\textwidth]{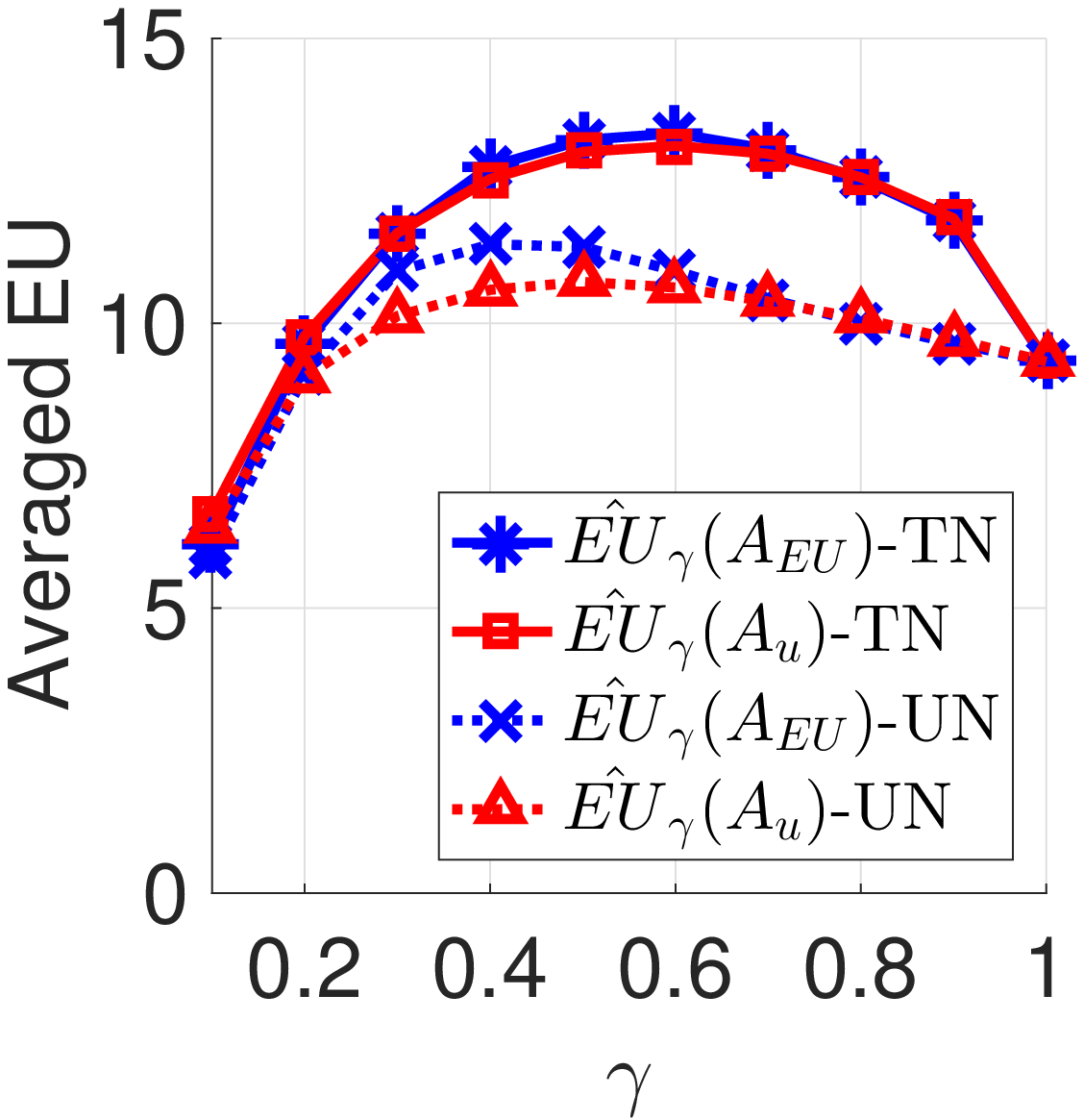}
                \caption{$n=30$, $\kappa=8$, $\Delta c = 0.5$}
                \label{fig:avg_30_users_kappa_8_delta_0_5}
        \end{subfigure}
        \begin{subfigure}[b]{0.45\columnwidth}
                \includegraphics[width=1\textwidth]{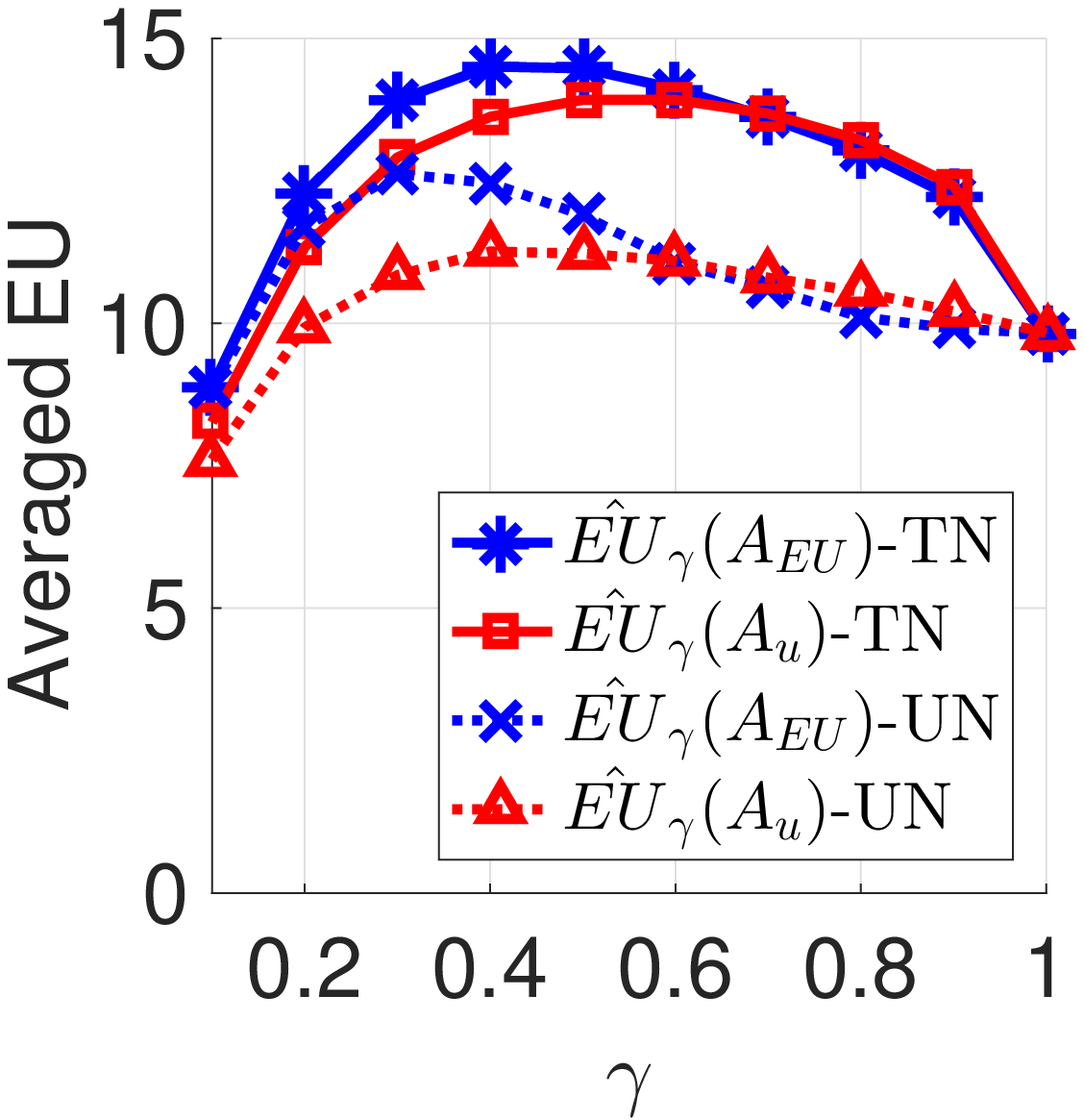}
                \caption{$\boldsymbol{n=60}$, $\kappa=8$, $\Delta c = 0.5$}
                \label{fig:avg_60_users_kappa_8_delta_0_5}
        \end{subfigure}
        \begin{subfigure}[b]{0.45\columnwidth}
                \includegraphics[width=1\textwidth]{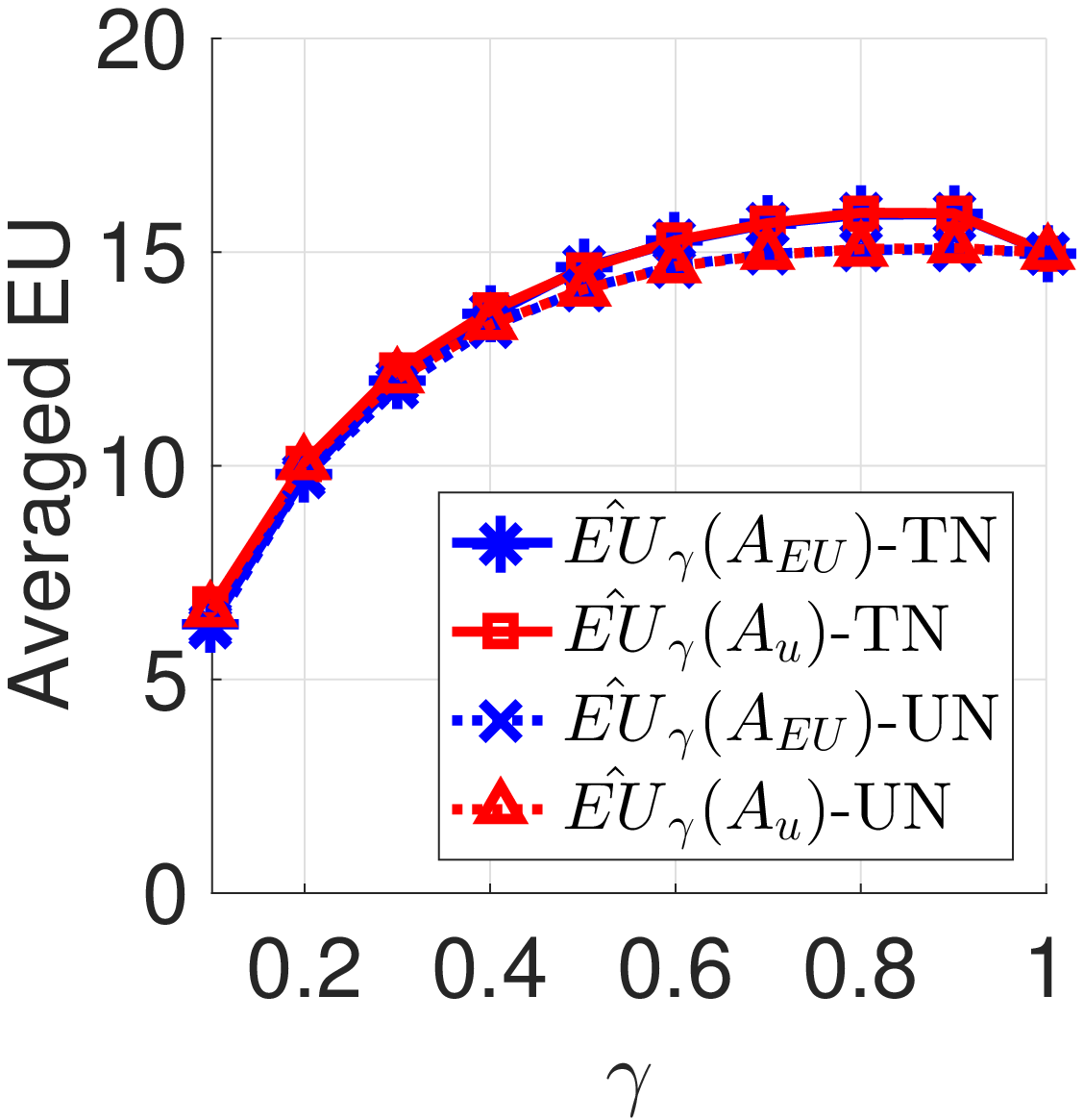}
                \caption{$n=30$, $\kappa=8$, $\boldsymbol{\Delta c = 0.1}$}
                \label{fig:avg_30_users_kappa_8_delta_0_1}
        \end{subfigure}
        \begin{subfigure}[b]{0.45\columnwidth}
                \includegraphics[width=1\textwidth]{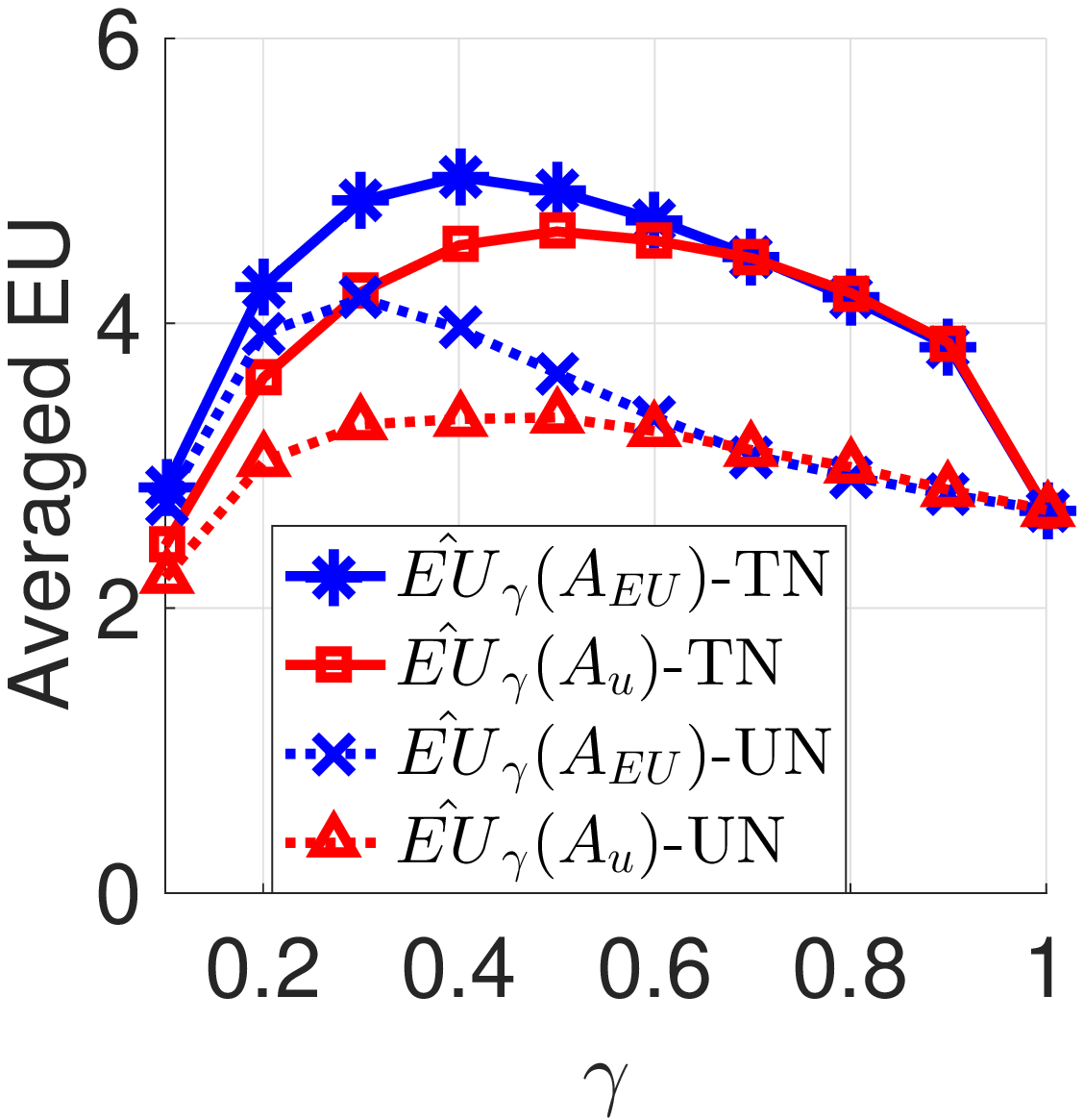}
                \caption{$n=30$, $\boldsymbol{\kappa=4}$, $\Delta c = 0.5$}
                \label{fig:avg_30_users_kappa_4_delta_0_5}
        \end{subfigure}
        \caption{Average EU achieved by USM-u and USM-EU for UN and TN cost distributions under different settings.}
        \label{fig:USM_u_vs_USM_EU_avg}
\end{figure}


In Fig.~\ref{fig:avg_30_users_kappa_8_delta_0_5}, \rev{we first observe that $\gamma^*$ achieving the maximum averaged EU is less than $1$ for both USM-EU and USM-u under both UN and TN distributions}. 
Intuitively, with a smaller $\gamma$, the platform can save money per user and send out more offers.
Although each offer is less likely to be accepted, the platform achieves a greater EU on average.
Second, at $\gamma=\gamma^*$, USM-EU achieves a higher EU than USM-u, especially for UN.
Besides, the fact that more money is saved per user under TN than UN with the same $\gamma$ explains the observation that both USM-EU and USM-u achieve a larger EU under TN than UN.

Similar behaviors are observed in Fig.~\ref{fig:avg_60_users_kappa_8_delta_0_5} and \ref{fig:avg_30_users_kappa_4_delta_0_5}.
But in Fig.~\ref{fig:avg_30_users_kappa_8_delta_0_1} when $\Delta c$ is changed from $0.5$ to $0.1$, the uncertainty in opportunity costs is reduced and energy costs become more dominant. 
In this case, the platform will not save much per user with a small $\gamma$ and should choose a larger $\gamma$.
Besides, we do not observe the advantage of USM-EU.

Furthermore, as shown in Fig.~\ref{fig:avg_30_users_kappa_8_delta_0_5}, USM-EU is slightly better than USM-u with $\gamma$ between 0.3 and 0.6 under UN, but their performance is very close for $\gamma\leq 0.2$ under UN and for all $\gamma$ values under TN. 
This is mainly because in those cases, both mechanisms will send out offers to more or all users (30 max. in this setting) and thus achieve very close EU.
When there are more users (Fig.~\ref{fig:avg_60_users_kappa_8_delta_0_5}), USM-EU is better than USM-u under both UN and TN \rev{for $\gamma < 0.6$}.
With smaller $\Delta c$ (Fig.~\ref{fig:avg_30_users_kappa_8_delta_0_1}), USM-EU is less advantageous. 
With smaller $\kappa$ (Fig.~\ref{fig:avg_30_users_kappa_4_delta_0_5}), each user is less valuable and USM-u selects fewer users, since it assumes no \rev{expired or rejected offers}.
In contrast, USM-EU considers the average-case utility and is more aggressive in user selection, which explains its better performance than USM-u with the same small $\gamma$. 

\subsection{Batched Offering vs. Sequential Offering}\label{sec:UM_PA_vs_OP}
In this simulation, we compare the following mechanisms: 
\rev{
\begin{itemize}
\item SB-u/EU: single-batch offering with USM-u or USM-EU;
\item MB-u/EU: multi-batch offering with USM-u or USM-EU;
\item SE: sequential offering.
\end{itemize}
}

\begin{figure}[ht!]
\centering
\begin{subfigure}[b]{0.45\columnwidth}
        \includegraphics[width=1\textwidth]{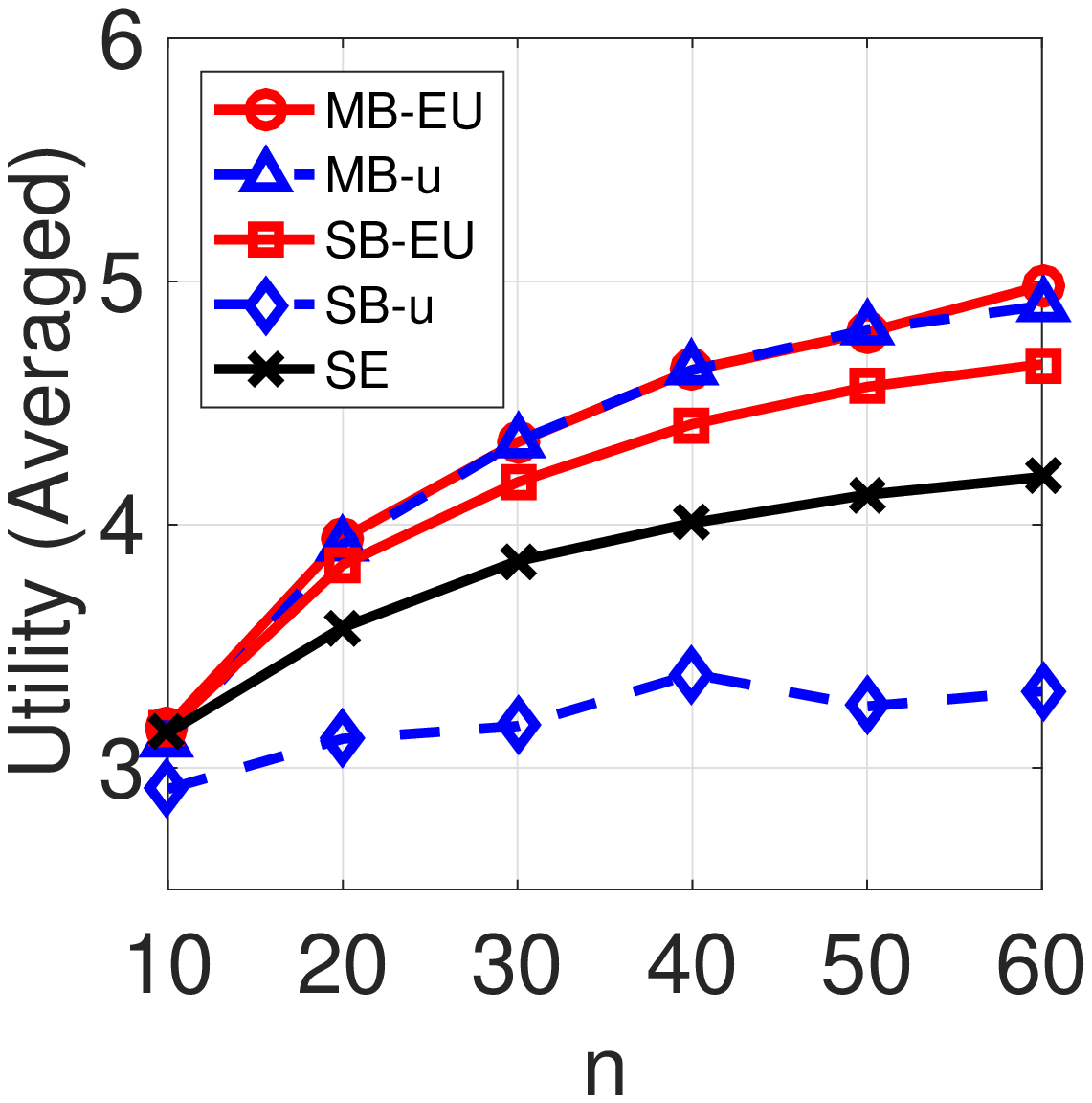}
        \caption{}
        \label{fig:impact_of_n_uniform}
\end{subfigure}
\begin{subfigure}[b]{0.45\columnwidth}
        \includegraphics[width=.93\textwidth]{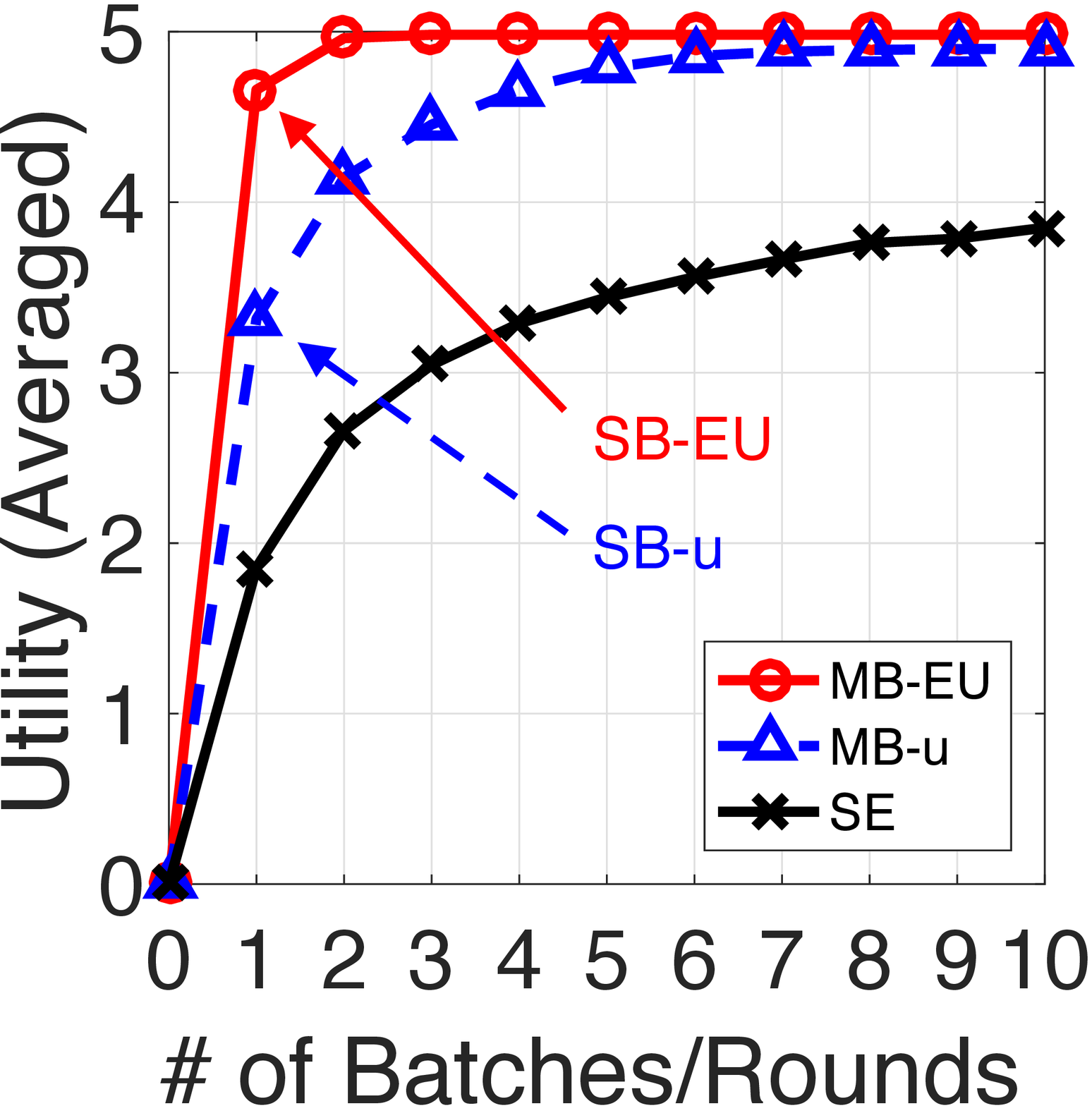}
        \caption{}
        \label{fig:utility_vs_num_of_batches_uniform}
\end{subfigure}
\begin{subfigure}[b]{0.45\columnwidth}
        \includegraphics[width=1\textwidth]{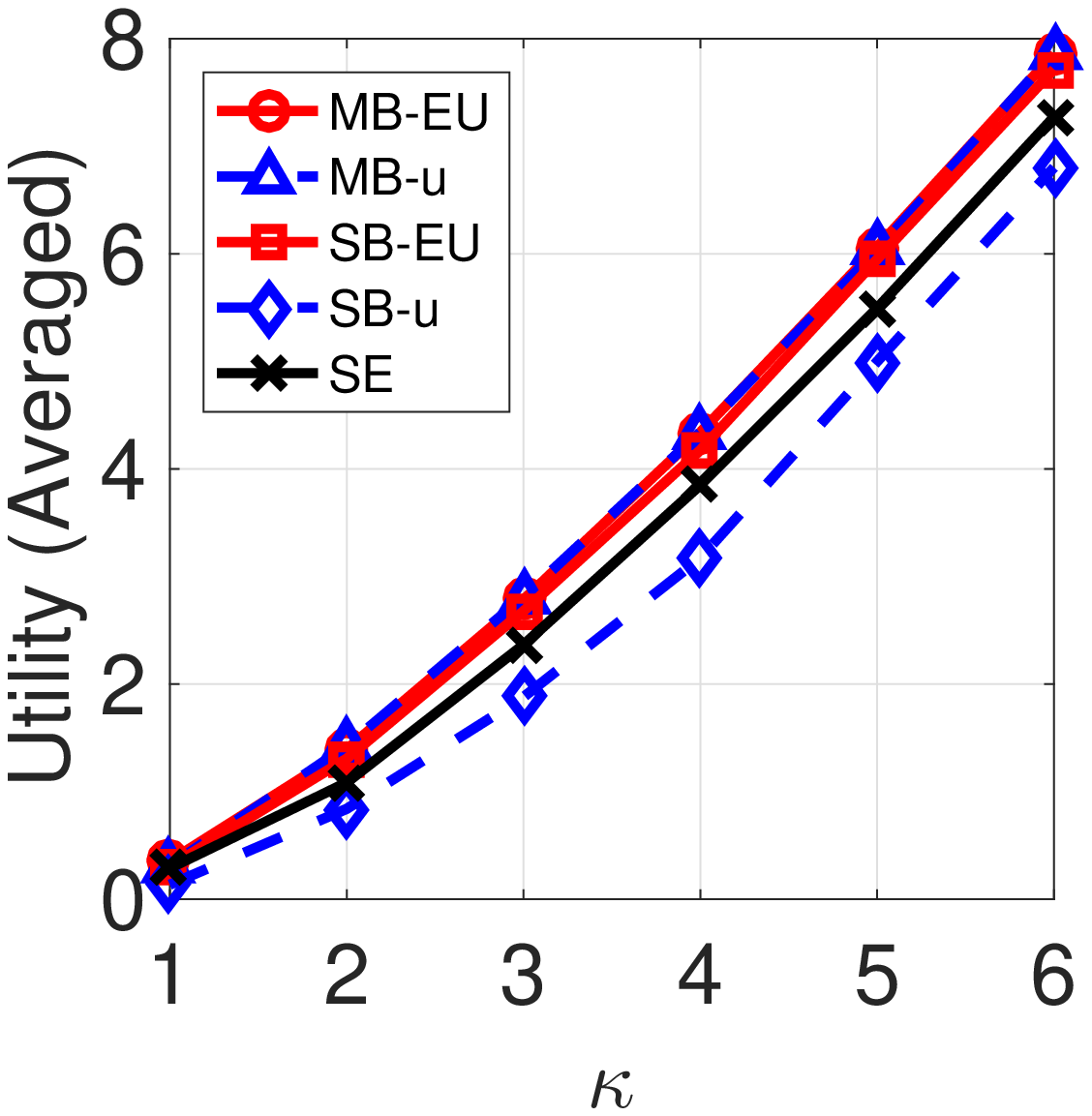}
        \caption{}
        \label{fig:impact_of_kappa_uniform}
\end{subfigure}
\begin{subfigure}[b]{0.45\columnwidth}
        \includegraphics[width=1\textwidth]{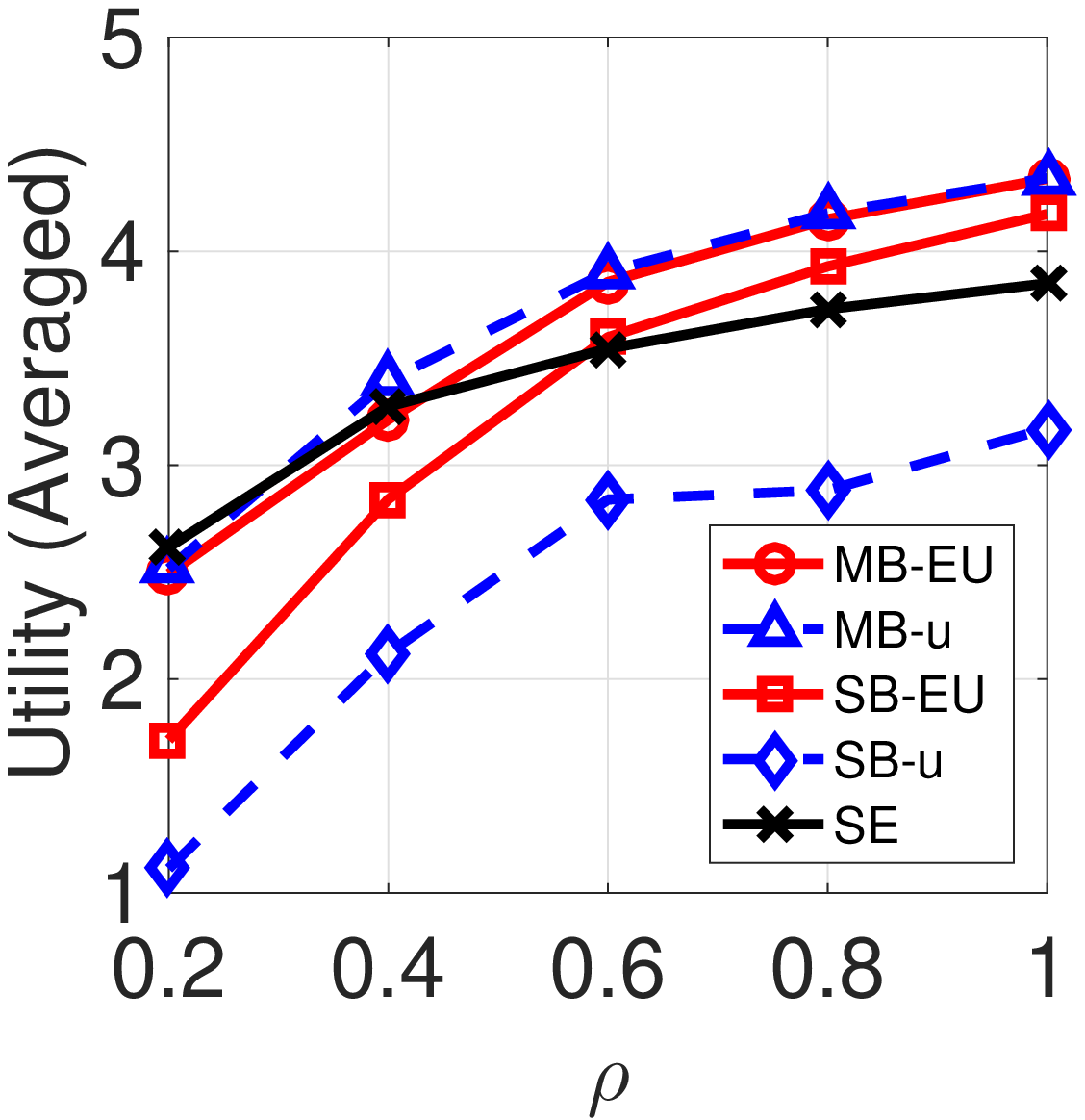}
        \caption{}
        \label{fig:impact_of_P_d_uniform}
\end{subfigure}
\caption{Comparison of SB-u, MB-u, SB-EU, MB-EU and SE.
        (a) Impact of $n$ ($\kappa=4$, no expired offers). 
        (b) Utility vs. number of batches of offers ($n=60$, $\kappa=4$, no expired offers).
        (c) Impact of $\kappa$ ($n=30$, no expired offers).
        (d) Impact of $\rho$ ($n=30$, $\kappa=4$, $\rho_i=\rho$ for each user $i$).}
\label{fig:comparison}
\end{figure}

We randomly select $n$ out of 60 users, and generate a set of noise variances and cost distributions, which are used in all iterations for each $n$. 
In the $i$-th iteration, a different set of sensing costs is independently generated from the cost distributions and used across different mechanisms for fair comparison.
We set $\Delta c = 0.5$, and varied $n$, $\kappa$ or $\{\rho_i\}$ to study their impacts on the average utility achieved by the platform.
All results are averaged over $50$ iterations. 
\rev{Due to space limit, only results for UN distributions are reported, but similar observations exist for TN distributions.}

\subsubsection{Impact of $n$ (number of users)}
We first set $\kappa=\rev{4}$ and $\rho_i = 1$ for each user $i$ (i.e., no expired offers). 
$n$ is varied from $10$ to $60$, and \rev{results are provided in Fig.~\ref{fig:impact_of_n_uniform}.}
\rev{First, we observe that all mechanisms \rev{achieve a higher utility on average} as $n$ increases.
Second, if only one batch/round is allowed, USM-EU achieves the highest utility, since it accounts for possible expiration and rejection and thus makes more offers in the first batch. 
The improvement of USM-EU (i.e., SB-EU) over USM-u (i.e., SB-u) varies from $8.5\%$ with $n=10$ to $40.5\%$ with $n=60$ (Fig.~\ref{fig:impact_of_n_uniform}), which means the advantage of SB-EU (over SB-u) is more obvious with more users. 
SE performs the worst, since it only sends out one offer in the first round. 

When more batches are allowed (Fig.~\ref{fig:utility_vs_num_of_batches_uniform}), all mechanisms perform better, since it is always beneficial to send out more batches to make up for expired or rejected offers in the previous batch.
Since USM-EU is very generous in making offers in the first batch, following batches become less profitable.
If the maximum number of batches is unlimited, USM-u (i.e., MB-u) eventually achieves very close performance with USM-EU (i.e., MB-EU), but the price is a much larger cumulative delay.
For instance, when $n=60$, MB-EU and MB-u make $2.5$ and $7.7$ batches of offers on average. 
The number is $24.9$ for SE, which is the worst.}

\subsubsection{Impact of $\kappa$ (currency in Eq.~(\ref{eq:valuation_function}))} 
We then fix $n=30$ and vary $\kappa$ from 1 to 6.
\rev{We set $\rho_i = 1$ for each user $i$.}
As shown in Fig.~\ref{fig:impact_of_kappa_uniform}, the average utility obtained by each mechanism increases as $\kappa$ increases, because the platform values per unit MI (log-scaled) more and is able to recruit more users.
\rev{Besides, SB-EU is still better than SB-u for different $\kappa$, but its advantage is less obvious when $\kappa$ gets larger.
For instance, the improvement is $127.1\%$ with $\kappa = 1$, but reduces to $13.4\%$ with $\kappa = 6$.}
Moreover, multi-batch offering is better than single-batch offering for both USM-EU and USM-u, \rev{but the improvement is more significant for USM-u}.


\subsubsection{Impact of $\rho$ (probability of unexpired offers)}\label{sec:impact_of_rho}
We set $n=30$ and $\kappa=4$. 
For simplicity, we assume $\rho_i = \rho$ for each user $i$ and vary $\rho$ from $0.2$ to $1$.
Results are provided in Fig.~\ref{fig:impact_of_P_d_uniform}.
First, we observe that all mechanisms are adversely affected when $\rho$ decreases.
Even though the platform knows $\rho$ and adjusts the price as in Eq.~(\ref{eq:pricing_rule}) to achieve $\gamma$, i.e., increasing the acceptance probability to offset the high expiration probability, it implies higher prices for users and consequently reduced utility. 
As mentioned earlier in Section~\ref{sec:sequential_offering}, the best price for each user in SE does not depend on $\rho$, but the resulting marginal EU does. 
Nevertheless, SE observes the outcome of the previous offer and continues offering, which explains why it (as well as multi-batch offering) is more robust against offer expiration than single-batch offering.

\section{Conclusion}\label{sec:conclusion}
In this work, we designed a crowd-sensing system for spatial-statistics-based radio mapping and developed pricing mechanisms, i.e., sequential and batched offering, based on EU maximization.
\rev{
We conducted extensive simulations to evaluate proposed mechanisms.
Our results show that if only one batch is allowed, the proposed mechanism based on EU maximization is significantly better than the utility-maximization-based baseline mechanism.
If multiple batches are permitted (and the number of batches is unlimited), the proposed mechanism achieves close performance with the baseline mechanism, but requires much fewer batches and thus a much smaller delay.
Sequential offering works better than the single-batch baseline mechanism, but has a much larger cumulative delay. 
Offer expiration adversely affected all mechanisms, but sequential and batched offering are relatively more robust.}


\appendices
\section{}
\subsection{Proof of Lemma~\ref{lemma:v_submodular}}
\label{proof:v_submodular}
Consider two sets $A$, $B$ such that $A\subseteq B \subseteq S$ and any $i \in S\setminus B$.
Let $f(A)=1+MI'(A)$.
First, we know that $f(A)$ is submodular from the submodularity of $MI(A)$, i.e.,
\begin{align}
f(A\cup\{i\})-f(A)&=MI(A\cup\{i\})-MI(A)+\alpha \nonumber \\
&\geq MI(B\cup\{i\})-MI(B)+\alpha \nonumber \\
&= f(B\cup\{i\})-f(B).
\end{align}
We also know that $f(\mathcal{A})$ is monotone, since $MI(A)$ is $\alpha$-approximately monotone and thus 
\begin{equation}
f(A\cup\{i\})-f(A)=MI(A\cup\{i\})-MI(A)+\alpha\geq 0.
\end{equation}

Let $a=f(A)$, $b=f(A\cup\{i\})$, $c=f(B)$, and $d=f(B\cup\{i\})$.
From the submodularity and monotonicity of $f(\cdot)$, we have $b-a\geq d-c \geq 0$ and $d\geq b$.
Let $a'=b-(d-c) \geq a$. 
Since $\log(\cdot)$ is non-decreasing concave, we have $\log(b)-\log(a)\geq \log(b)-\log(a')\geq \log(d)-\log(c) \geq 0$.
Hence, $v(A\cup\{i\})-v(A) \geq v(B\cup\{i\})-v(B)\geq 0$, and $v(\cdot)$ is submodular monotone.

\subsection{Proof of Lemma~\ref{lemma:EU_submodular}}\label{proof:EU_submodular}
We notice that $u(A_{\bf y}, {\bf p})$ is submodular in $A_{\bf y}$ given ${\bf p}$, since $v(\cdot)$ is submodular (Lemma~\ref{lemma:v_submodular}) and
\begin{align}
u(A_{\bf y} \cup \{i\}, {\bf p}) & - u(A_{\bf y}, {\bf p}) = v(i|A_{\bf y} ) - p_i \geq v(i|B) - p_i \nonumber \\
&\geq u(B \cup \{i\}, {\bf p})-u(B, {\bf p})
\end{align}
for $A_{\bf y} \subseteq B \subseteq S$ and any $i \in S \setminus B$.
Since the class of submodular functions are closed under taking expectations, it follows that $EU(A, {\bf p})$ is submodular in $A$ given $ {\bf p}$.


\subsection{Proof of Theorem~\ref{theorem:USM_with_approx_EU_bound}} \label{proof:USM_with_approx_EU_bound}
Before proving Theorem~\ref{theorem:USM_with_approx_EU_bound}, we first prove the following lemma. 
\begin{mylemma}\label{lemma:USM_with_approx_EU_bound0}
Given a nonnegative submodular function $f:2^S \mapsto \mathbb{R}_+$ and its estimate $\hat{f}$ with $|\hat{f}(A)-f(A)| \leq \epsilon$ for any $A \subseteq S$ and some small $\epsilon>0$, {\tt USM}($S$, $\hat{f}$) returns a solution $A$ with the following performance guarantee, 
\begin{equation}
f(A) \geq \frac{1}{3}f(OPT) - \frac{1}{3}(2n + 2)\epsilon
\end{equation}
where $OPT = \arg \max_{A \subseteq S} f(A)$ and $n=|S|$.
\end{mylemma}
\begin{proof}
Our proof is inspired by the proof in \cite{buchbinder2015tight}.
Let us start with Lemma~\ref{lemma:USM_with_approx_EU_bound1}.

\begin{mylemma}\label{lemma:USM_with_approx_EU_bound1}
For every $1 \leq i \leq n$, $a_i + b_i \geq - 4 \epsilon$, where $a_i = \hat{f}(A_{i-1} \cup \{u_i\}) - \hat{f}(A_{i-1})$ and $b_i = \hat{f}(B_{i-1} \setminus \{u_i\}) - \hat{f}(B_{i-1})$.
\end{mylemma}
\begin{proof}
Since $|\hat{f}(A)-f(A)| \leq \epsilon$, $\forall A \subseteq S$, we have
\begin{equation}\label{eq:range_of_estimated_f}
f(A) - \epsilon \leq \hat{f}(A) \leq f(A) + \epsilon, ~\forall A \subseteq S.
\end{equation}
Notice that $(A_{i-1} \cup \{u_i\}) \cup (B_{i-1} \setminus \{u_i\}) = B_{i-1}$, $(A_{i-1} \cup \{u_i\}) \cap (B_{i-1} - u_i) = A_{i-1}$. 
Based on both observations and submodularity of $f$, we get
\begin{align}
a_i + b_i &= [\hat{f}(A_{i-1} \cup \{u_i\}) - \hat{f}(A_{i-1})] \nonumber \\
&~~~~~~~ + [\hat{f}(B_{i-1} \setminus \{u_i\}) - \hat{f}(B_{i-1})] \\
&\geq [f(A_{i-1} \cup \{u_i\}) + f(B_{i-1} \setminus \{u_i\})] \nonumber \\
&~~~~~~~ - [f(A_{i-1}) + f(B_{i-1})] - 4\epsilon \geq -4 \epsilon.
\end{align}\end{proof}
Define $OPT_i \triangleq (OPT \cup A_i) \cap (B_i)$.
Thus, $OPT_0 = OPT$ and the algorithm outputs $OPT_n = A_n = B_n$.
Examine the sequence $f(OPT_0)$, ..., $f(OPT_n)$, which starts with $f(OPT)$ and ends with the $f$ value of the output of the algorithm.
The idea is to bound the total loss of value along this sequence. 

\begin{mylemma}\label{lemma:USM_with_approx_EU_bound2}
For every $1 \leq i \leq n$, we have 
\begin{align}
f(OPT_{i-1})-f(OPT_i) &\leq [\hat{f}(A_i)-\hat{f}(A_{i-1})] \nonumber \\
&~~+ [\hat{f}(B_i) - \hat{f}(B_{i-1})] + 2 \epsilon.
\end{align}
\end{mylemma}
\begin{proof}
W.L.O.G., we assume that $a_i\geq b_i$, i.e., $A_i \leftarrow A_{i-1} \cup \{u_i\}$, $B_i \leftarrow B_{i-1}$ (the other case is analogous).
Notice that in this case $OPT_i = (OPT \cup A_i) \cap B_i = OPT_{i-1} \cup \{u_i\}$, $B_i = B_{i-1}$ and $\hat{f}(B_i) = \hat{f}(B_{i-1})$.
Hence, the inequality we need to prove is that 
\begin{align}
f(OPT_{i-1}) &- f(OPT_{i-1} \cup \{u_i\}) \nonumber \\
&\leq [\hat{f}(A_i)-\hat{f}(A_{i-1})] + 2\epsilon = a_i + 2\epsilon
\end{align}
We now consider two cases. 
If $u_i \in OPT$, then the left-hand of the inequality is 0, and all we need to show is that $a_i \geq -2\epsilon$.
This is true since $a_i + b_i \geq - 4\epsilon$ by Lemma~\ref{lemma:USM_with_approx_EU_bound1}, and we assumed $a_i \geq b_i$.

If $u_i \notin OPT$, then also $u_i \notin OPT_{i-1}$, and thus
\footnotesize
\begin{align}
f(&OPT_{i-1}) - f(OPT_{i-1} \cup \{u_i\}) \leq f(B_{i-1} \setminus \{u_i\}) - f(B_{i-1}) \nonumber \\
&\leq \hat{f}(B_{i-1} \setminus \{u_i\}) - \hat{f}(B_{i-1}) + 2\epsilon = b_i + 2\epsilon \leq a_i + 2\epsilon.
\end{align}
\normalsize
The first inequality follows by submodularity: $OPT_{i-1} = ((OPT \cup A_{i-1}) \cap B_{i-1}) \subseteq (B_{i-1} \setminus \{u_i\})$ (recall that $u_i \in B_{i-1}$ and $u_i \notin OPT_{i-1}$).
The second and third inequalities follow from Eq.~(\ref{eq:range_of_estimated_f}) and our assumption that $a_i \geq b_i$, respectively.
\end{proof}

Summing up Lemma~\ref{lemma:USM_with_approx_EU_bound2} for every $1\leq i \leq n$, we have
\footnotesize
\begin{align}
&\sum_{i=1}^n [f(OPT_{i-1}) - f(OPT_i)] \nonumber \\
&\leq \sum_{i=1}^n [\hat{f}(A_i)-\hat{f}(A_{i-1})] +  \sum_{i=1}^n [\hat{f}(B_i) - \hat{f}(B_{i-1})] + 2 n \epsilon.
\end{align}
\normalsize

The above sum is telescopic and we have
\footnotesize
\begin{align}
f(OPT_{0}) &- f(OPT_n) \leq [\hat{f}(A_n)-\hat{f}(A_{0})] +  [\hat{f}(B_n) - \hat{f}(B_0)] + 2 n \epsilon \nonumber \\
& \leq \hat{f}(A_n) + \hat{f}(B_n) + 2 n \epsilon \leq f(A_n) + f(B_n) + (2n + 2) \epsilon
\end{align}
\normalsize
By our definition, $OPT_0 = OPT$ and $OPT_n = A_n = B_n$.
Then we obtain that $f(OPT) \leq 3 f(A_n) + (2n+2)\epsilon$ 
and $f(A_n)=f(B_n) \geq \frac{1}{3}f(OPT) - \frac{1}{3}(2n + 2)\epsilon$.
\end{proof}

Now let us prove Theorem~\ref{theorem:USM_with_approx_EU_bound}. 
By Lemma~\ref{lemma:USM_with_approx_EU_bound0}, we have $EU_\gamma'(A) \geq \frac{1}{3}EU_\gamma'(OPT) - \frac{1}{3}(2n + 2)\epsilon$, where $OPT = \arg \max_{A \subseteq S} EU_\gamma'(A)$, which is also the optimal solution to $\max_{A \subseteq S} EU_\gamma(A)$. 
By definition of $EU_\gamma'(A)$, we obtain $EU_\gamma(A) + p_0 \geq \frac{1}{3}[EU_\gamma(OPT) + p_0] - \frac{1}{3} (2n + 2)\epsilon$ and $EU_\gamma(A) \geq \frac{1}{3} EU_\gamma(OPT) - \frac{2}{3} p_0 - \frac{1}{3} (2n+2)\epsilon$.

\bibliographystyle{IEEEtran}
\bibliography{IEEEabrv}

%

\begin{IEEEbiography}[{\includegraphics[width=1in,height=1.25in,clip,keepaspectratio]{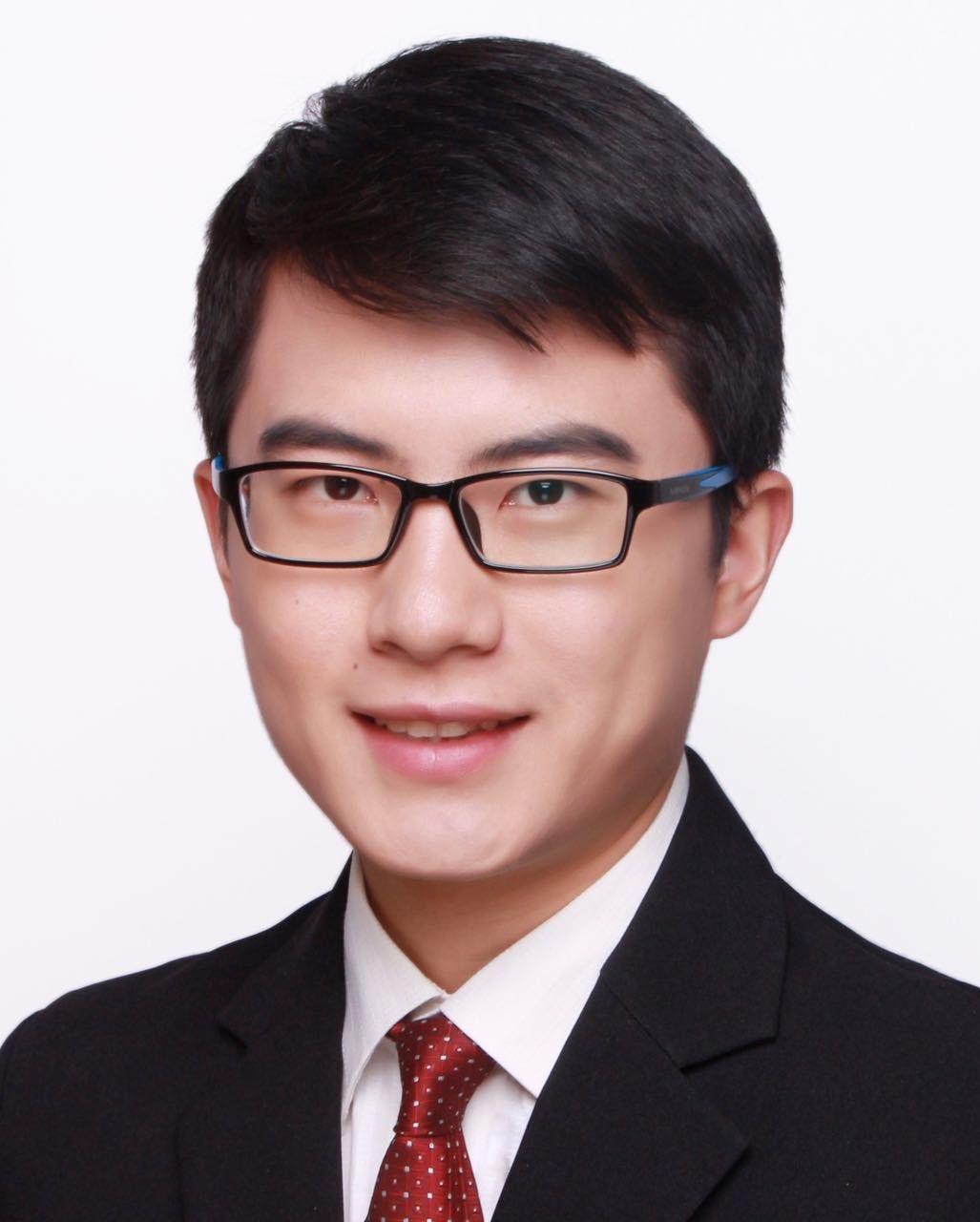}}]{Xuhang Ying}
(S'15) received the B.Eng. degree in Information Engineering from the Chinese University of Hong Kong in 2013, and the M.S. degree in Electrical Engineering in the University of Washington in 2016, where he is currently pursuing the Ph.D. degree in Electrical Engineering. 
His research has focused on incentivized crowd-sensing, resource allocation, spectrum enforcement in shared bands. 
\end{IEEEbiography}

\begin{IEEEbiography}[{\includegraphics[width=1in,height=1.25in,clip,keepaspectratio]{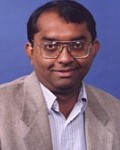}}]{Sumit Roy}
(S'84--M'88--SM'00--F'07) received the
B.Tech. degree in electrical engineering from the Indian Institute of Technology, Kanpur, India, in 1983 and the M.S. and Ph.D. degrees in electrical engineering and the M.A. degree in statistics and applied probability from the University of California Santa Barbara, CA, USA, in 1985, 1988, and 1988, respectively. 
He spent 2001--2003 on academic leave at Intel Wireless Technology Laboratory as a Senior Researcher engaged in research and standards development for ultrawideband systems (wireless PANs) and next-generation high-speed wireless LANs. 
He is currently a Professor of electrical engineering at the University of Washington, Seattle, WA, USA, where his research interests include analysis/design of communication systems/networks, with an emphasis on next generation mobile/wireless and sensor networks. 
He served as a Science Foundation of Ireland Isaac Walton Fellow during a sabbatical at the University College, Dublin, Ireland (January--June 2008), and was the recipient of a Royal Academy of Engineering (U.K.)
Distinguished Visiting Fellowship during summer 2011. 
He was elevated to IEEE Fellow by the Communications Society in 2007 for his contributions to multiuser communications theory and cross-layer design of wireless networking standards. 
\end{IEEEbiography}

\begin{IEEEbiography}[{\includegraphics[width=1in,height=1.25in,clip,keepaspectratio]{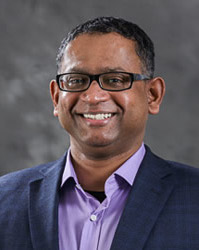}}]{Radha Poovendran}
(F'15) is a Professor and the Chair of the Electrical Engineering Department, Director of the Network Security Lab (NSL), and Associate Director of the Center for Excellence in Information Assurance Research and Education, all at the University of Washington - Seattle. 
He received the B.S. degree in Electrical Engineering and the M.S. degree in Electrical and Computer Engineering from the Indian Institute of Technology- Bombay and University of Michigan - Ann Arbor in 1988 and 1992, respectively. He received the Ph.D. degree in Electrical and Computer Engineering from the University of Maryland - College Park in 1999. 
His research interests are in the areas of wireless and sensor network security, control and security of cyber-physical systems, adversarial modeling, smart connected communities, control-security, gamessecurity and information theoretic security in the context of wireless mobile networks. 
He was elected a Fellow of the IEEE for his contributions to security in cyber-physical systems. 
He is a recipient of the NSA LUCITE Rising Star Award (1999), National Science Foundation CAREER (2001),
ARO YIP (2002), ONR YIP (2004), and PECASE (2005) for his research contributions to multi-user wireless security. 
He is also a recipient of the Outstanding Teaching Award and Outstanding Research Advisor Award from the University of Washington EE (2002), Graduate Mentor Award from Office of the Chancellor at University of California - San Diego (2006), and the University of Maryland ECE Distinguished Alumni Award (2016). 
He was co-author of award-winning papers including IEEE/IFIP William C. Carter Award Paper (2010) and WiOpt Best Paper Award (2012).
\end{IEEEbiography}

\end{document}